\documentclass[final,1p]{elsarticle}
\journal{arXiv, Nuclear Physics B format}

\usepackage{amssymb}
\usepackage{amsmath}
\usepackage{amsthm}
\usepackage{graphicx}
\usepackage{epstopdf}
\usepackage{comment}
\usepackage{float}

\newcommand{\diag}{\operatorname{diag}}

\newcommand{\GeV}{\text{ GeV}}
\newcommand{\MeV}{\text{ MeV}}

\newcommand{\gp}{g\,'}
\newcommand{\gpt}{g\,'^{\,2}}
 \newcommand{\pp}{p\,'}
\newcommand{\qp}{q\,'}
\newcommand{\rp}{r\,'}
\newcommand{\phip}{\phi\,'}
\newcommand{\thetap}{\theta\,'}

\newtheorem{definition}{Definition}
\newtheorem{theorem}{Theorem}

\newtheorem{corollary}{Corollary}
\newcommand{\req}[1]{Eq.\,(\ref{#1})}

\begin{document} 

\begin{frontmatter}

\title{Relic Neutrino Freeze-out: Dependence on Natural Constants}
\author{Jeremiah Birrell$^{a,b}$, Cheng Tao Yang$^{b,c}$,  and Johann Rafelski$^{b}$}
\address{ $^a$Program in Applied Mathematics, The University of Arizona, Tucson,  Arizona, 85721, USA\\[0.2cm]
 $^{b}$Department of Physics, The University of Arizona,  Tucson,  Arizona, 85721, USA\\[0.2cm]
 $^{c}$Department of Physics and Graduate Institute of Astrophysics,\\ National Taiwan University, Taipei, Taiwan, 10617}

\date{June 6, 2014}

\begin{abstract}
Analysis of  cosmic microwave background radiation fluctuations favors an effective number of neutrinos, $N_\nu>3$. This motivates a reinvestigation of the neutrino freeze-out process.  Here we characterize the dependence of  $N_\nu$ on the Standard Model (SM) parameters that govern neutrino freeze-out. We show that $N_\nu$ depends on a combination $\eta$  of several natural constants characterizing the relative strength of weak interaction processes in the early Universe and on the  Weinberg angle $\sin^2\theta_W$. We determine numerically the dependence  $N_\nu(\eta,\sin^2\theta_W)$ and discuss these results. The extensive numerical computations are made possible by two novel numerical procedures:  a spectral method Boltzmann equation solver adapted to allow for strong reheating and emergent chemical non-equilibrium, and a  method to evaluate Boltzmann equation collision integrals that generates a smooth integrand.
\end{abstract}


\end{frontmatter}

\section{Introduction}\label{Intro}
The relic neutrino background is believed to be a  well preserved probe of a Universe only a second old. The properties of the neutrino background are influenced by the details of the freeze-out or decoupling process at a temperature $T={\cal O}(1\MeV)$, which is in turn controlled by the  standard model (SM) of particle physics  parameters.  In this paper, we study the influence of SM parameters on the neutrino distribution after freeze-out.  This exercise is of interest because:
\begin{itemize}
\item  
There is a tension between the known number of neutrinos (three flavors) and  the effective number of neutrinos $N_\nu\simeq 3.5$ deduced from the study of the cosmic microwave background (CMB).  Detailed analysis of the CMB by the Planck satellite collaboration (Planck)~\cite{Planck} tests both gravitational and SM interactions in the early Universe. So far very little effort has been devoted to the understanding how these results characterize SM properties in the early Universe.
\item
The topic of time variation of natural constants is a very active field with a long history~\cite{Uzan:2010pm}.  Consideration of neutrino freeze-out dependence on natural constants provides new insights on the time and/or temperature variation of several SM parameters considered at an early era,  $t\approx1$s,  in the Universe's evolution. 
\end{itemize}


The presence of  relativistic particles, such as neutrinos, strongly impacts the dynamics of the expansion of the Universe,   constrained by direct measurement and analysis of cosmic microwave background (CMB) temperature fluctuations~\cite{Planck}. The effect is described by  $N_\nu$ which quantifies the amount of radiation energy density, $\rho_r$, in the Universe prior to photon freeze-out and after $e^\pm$ annihilation and is defined by
\begin{equation}\label{eq:Nnu}
\rho_r=(1+(7/8)R_\nu^{4}N_\nu)\rho_\gamma,
\end{equation}
where $\rho_\gamma$ is the photon energy density. The factor 7/8 is the ratio of Fermi to Bose normalization in $\rho$ and the neutrino to photon temperature ratio $R_\nu$ is the result of the transfer of $e^+e^-$ entropy into photons after Standard Model (SM) left handed neutrino freeze-out:
\begin{equation}\label{T_nu_T_gamma}
R_\nu\equiv T_\nu/T_\gamma, \qquad R_\nu^0 =({4}/{11})^{1/3}.
\end{equation}
This well known value  $R_\nu^0$ arises in the limit where no entropy from the annihilating $e^\pm$ pairs is transferred to neutrinos, i.e. all entropy feeds and reheats the photon background.

We expect to measure a  value  $N_\nu=3$, i.e. the number of SM left handed neutrino flavors  only if:
\begin{enumerate}
\item
Photons and SM neutrinos are the only effectively massless particle species in the Universe between the freeze-out of the left handed neutrinos at  $T_\gamma={\cal{O}}(1)$ MeV and photon freeze-out at $T_\gamma=0.25$ eV, and 
\item 
No flow of entropy from $e^\pm$ annihilation to neutrinos occurs. The current status is that computation of the neutrino freeze-out process employing SM two body scattering interactions and carried out using the Boltzmann equation gives $N_\nu^{\rm th}=3.046$~\cite{Mangano2005}, a value close to the number of flavors.  
\end{enumerate}
 
The value of $N_\nu$ can be  measured by fitting to observational data, such as the distribution of CMB temperature fluctuations. The  Planck~\cite{Planck}  analysis gives $N_\nu=3.36\pm0.34$ (CMB only), $N_\nu=3.30\pm 0.27$ (CMB+BAO), and $N_\nu=3.62\pm 0.25$ (CMB+$H_0$) ($68\%$ confidence levels).  Combinations of of Planck results with other priors are also reported by the Planck collaboration with most resulting in central values $N_\nu\in (3.3,3.6)$. With more dedicated CMB experiments on the drawing board it is believed that a significantly more precise value of $N_\nu$  is forthcoming in the next decade.

The tension between the values inferred from observation and the SM prediction has inspired various theories, including the consideration of:  modified neutrino interactions~\cite{Mangano:2006ar};  a model in which the temperature of decoupling was a model parameter~\cite{Birrell:2013_2};  a  model of a new spontaneously broken symmetry associated with massless Goldstone bosons that freeze out prior to the disappearance of muons~\cite{Weinberg:2013kea}, and   a similar consideration motivated by recognition that  the  quark-gluon plasma phase transition potentially offers the required physics context~\cite{Birrell:2014connect}.

In this paper we explore the dependence of $N_\nu$ on the value of  natural constants within realm of known interactions. We show that  $N_\nu$ depends {\em only} on the magnitude of the  Weinberg angle in the form $\sin^2\theta_W$, and  a dimensionless relative interaction strength parameter $\eta$,
\begin{equation}\label{eq:etaMp}
\eta\equiv M_p m_e^3 G_F^2, \qquad M_p^2\equiv \frac{1}{8\pi G_N}, 
\end{equation}
a combination of  the electron mass $m_e$, Newton constant $G_N$, and the Fermi constant $G_F$.  The magnitude of  $\sin^2\theta_W$   is not fixed within the SM and  could be subject to variation as function of time or temperature. We find that   $\sin^2\theta_W$  which  differs substantially from the value measured today in vacuum is capable of significantly altering the value of $N_\nu$ that is generated during neutrino freeze-out. We further show   the combined effect of modified $\eta$ and $\sin^2\theta_W$ and argue that this can remove or at least reduce the tension of $N_\nu$ with the Planck data.

In section \ref{sec:boltz_param} we introduce the two body scattering description of neutrino freeze-out. In particular we discuss the Boltzmann equation, which is used to model the neutrino freeze-out process, and present the matrix elements that control neutrino freeze-out in subsection \ref{ssec:EB_Eq}. We then discuss in subsections \ref{ssec:matrix}  and \ref{ssec:int_st} the dependence of the Boltzmann equation on SM parameters, the Weinberg angle $\sin^2\theta_W$,  and the interaction strength parameter $\eta$ respectively.  In section \ref{sec:math} we introduce the technical methods we use to solve the Boltzmann equation. In subsection \ref{sec:boltz_solve} we outline our solution method for the Boltzmann equation. In subsection \ref{sec:coll_simp} we detail a new method for analytically simplifying the collision integrals in order to reduce the numerical integration costs. In subsection \ref{sec:comp} we compare with the results  by previous authors, highlighting the improvements we have made.

In section \ref{sec:param_dep} we show the impact of SM parameter values on neutrino freeze-out and the effective number of neutrinos.  In particular, we   show how the impact of the strength parameter $\eta$  and $\sin^2\theta_W$ on $N_\nu$.  We discuss the implications and connections of this work to other areas of physics, namely Big Bang nucleosynthesis and dark radiation, in section \ref{sec:disc}. We give our concluding analysis in section \ref{sec:concl}. \ref{app:delta} contains some mathematical background that is useful for the collision integral calculations of subsection \ref{sec:coll_simp} and  in \ref{app:nu_matrix_elements} we apply the method  to the processes involved in neutrino freeze-out. Finally, \ref{app:Tups} contains additional plots and numerical fits that show the impact of SM parameters on various quantities characterizing the neutrino distributions after freeze-out.

\section{Dynamical Description of Neutrino Freeze-out}\label{sec:boltz_param}
\subsection{Einstein-Boltzmann Equation}\label{ssec:EB_Eq}
To model the flow of energy and entropy into the relic neutrino distribution, and hence obtain the value of $N_\nu$ after freeze-out, we must solve the general relativistic Einstein-Boltzmann equation.  Several references  discuss this generalization of the Boltzmann equation in detail~\cite{andre,cercignani,bruhat,ehlers,kolb,bernstein2004kinetic} and  other works  specialize  this to the question of neutrino freeze-out~\cite{Madsen,Dolgov_Hansen,Gnedin,Esposito2000,Mangano2002,Mangano2005}.  Here we provide a quick overview of this literature. In the context of  general relativity the  Boltzmann equation is given by
\begin{equation}\label{boltzmann}
p^\alpha\partial_{x^\alpha}f-\Gamma^j_{\mu\nu}p^\mu p^\nu\partial_{p^j}f=C[f].
\end{equation}
$\Gamma^j_{\mu\nu}$ are the Christoffel symbols and so the left hand side expresses the fact that particles undergo geodesic motion in between point collisions.

The term $C[f]$ on the right hand side of the Boltzmann equation is called the collision operator and models the short range scattering processes that cause deviations from geodesic motion. For $2\leftrightarrow 2$ reactions between fermions, such as neutrinos and $e^\pm$, the collision operator takes the form
\begin{align}\label{coll}
C[f_1]=&\frac{1}{2}\int F(p_1,p_2,p_3,p_4) S |\mathcal{M}|^2(2\pi)^4\delta(\Delta p)\prod_{i=2}^4\delta_0(p_i^2-m_i^2)\frac{d^4p_i}{(2\pi)^3},\\
F=&f_3(p_3)f_4(p_4)f^1(p_1)f^2(p_2)-f_1(p_1)f_2(p_2)f^3(p_3)f^4(p_4),\notag\\
f^i=&1- f_i.\notag
\end{align}
Here $|\mathcal{M}|^2$ is the process amplitude or matrix element, $S$ is a numerical factor that incorporates symmetries and prevents over-counting, $f^i$ are the fermi blocking factors, $\delta(\Delta p)$ enforces four-momentum conservation in the reactions, and the $\delta_0(p_i^2-m_i^2)$ restrict the four momenta to the future timelike mass shells.

We now restrict our attention to systems of fermions under the assumption of homogeneity and isotropy. We assume that the particles are effectively massless,  i.e. the temperature is much greater than the mass scale.  Homogeneity and isotropy imply that the distribution function of each particle species under consideration has the form $f=f(t,p)$ where $p$ is the magnitude of the spacial component of the four momentum.  In a spatially flat FRW universe the Boltzmann equation reduces to
\begin{equation}\label{boltzmann_p}
\partial_t f-pH \partial_p f=\frac{1}{E}C[f],\hspace{2mm} H\equiv\frac{\dot{a}}{a}.
\end{equation}
This, combined with the Einstein equations and the matrix elements for the relevant processes, constitutes the dynamical equations governing neutrino freeze-out. 

We also obtain formulas for the rate of change in the  number density and energy density of the $i$th species
\begin{align}\label{n_div}
\frac{1}{a^3}\frac{d}{dt}(a^3n_i)=&\frac{g_p}{(2\pi)^3}\int C[f_i] \frac{d^3p}{E}.\\
\label{rho_div}
\frac{1}{a^4}\frac{d}{dt}(a^4\rho_i)=&\frac{g_p}{(2\pi)^3}\int C[f_i] d^3p .
\end{align} 
For free-streaming particles the vanishing of the collision operator implies conservation of `comoving' particle number of the $i$th species. From the associated powers of $a$ in \req{n_div} and \req{rho_div} we see that  the  energy of a free streaming particle  scales as $1/a$. This means that the distribution  of a free streaming massive particle species will, once the mass scale becomes relevant, evolve into non-thermal shape~\cite{Birrell:2013_2,nu_today}. 

The  matrix elements for weak force scattering processes involving neutrinos and $e^\pm$ are given in tables \ref{table:nu_e_reac} and \ref{table:nu_mu_reac}.  They were obtained from Ref.\cite{Dolgov_Hansen} and are valid in the limit $|p|\ll M_W,M_Z$, where in vacuum the gauge boson masses are $M_W=80.4\,{\rm GeV},\ M_Z=91.19\,{\rm GeV}$.

\begin{table}[h]
\centering 
\begin{tabular}{|c|c|}
\hline
Process &$S|\mathcal{M}|^2$  \\
\hline
$\nu_e+\bar\nu_e\rightarrow\nu_e+\bar\nu_e$ & $128G_F^2(p_1\cdot p_4)(p_2\cdot p_3)$\\
\hline
$\nu_e+\nu_e\rightarrow\nu_e+\nu_e$ & $64G_F^2(p_1\cdot p_2)(p_3\cdot p_4)$\\
\hline
$\nu_e+\bar\nu_e\rightarrow\nu_j+\bar\nu_j$&$32G_F^2(p_1\cdot p_4)(p_2\cdot p_3)$\\
\hline
$\nu_e+\bar\nu_j\rightarrow\nu_e+\bar\nu_j$ & $32G_F^2(p_1\cdot p_4)(p_2\cdot p_3)$\\
\hline
$\nu_e+\nu_j\rightarrow\nu_e+\nu_j$&$32G_F^2(p_1\cdot p_2)(p_3\cdot p_4)$\\
\hline
$\nu_e+\bar\nu_e\rightarrow e^++e^-$ & $128G_F^2[g_L^2(p_1\cdot p_4)(p_2\cdot p_3)+g_R^2(p_1\cdot p_3)(p_2\cdot p_4)+g_Lg_Rm_e^2(p_1\cdot p_2)]$\\
\hline
$\nu_e+e^-\rightarrow\nu_e+e^-$ & $128G_F^2[g_L^2(p_1\cdot p_2)(p_3\cdot p_4)+g_R^2(p_1\cdot p_4)(p_2\cdot p_3)-g_Lg_Rm_e^2(p_1\cdot p_3)]$\\
\hline
$\nu_e+e^+\rightarrow\nu_e+e^+$ & $128G_F^2[g_R^2(p_1\cdot p_2)(p_3\cdot p_4)+g_L^2(p_1\cdot p_4)(p_2\cdot p_3)-g_Lg_Rm_e^2(p_1\cdot p_3)]$\\
\hline
\end{tabular}
\caption{Matrix elements~\cite{Dolgov_Hansen} for electron neutrino processes where particle  $j=\mu,\tau$,  $g_L=\frac{1}{2}+\sin^2\theta_W$, $g_R=\sin^2\theta_W$.}
\label{table:nu_e_reac}
\end{table}

\begin{table}[h]
\centering 
\begin{tabular}{|c|c|}
\hline
Process &$S|\mathcal{M}|^2$  \\
\hline
$\nu_i+\bar\nu_i\rightarrow\nu_i+\bar\nu_i$ & $128G_F^2(p_1\cdot p_4)(p_2\cdot p_3)$\\
\hline
$\nu_i+\nu_i\rightarrow\nu_i+\nu_i$ & $64G_F^2(p_1\cdot p_2)(p_3\cdot p_4)$\\
\hline
$\nu_i+\bar\nu_i\rightarrow\nu_j+\bar\nu_j$&$32G_F^2(p_1\cdot p_4)(p_2\cdot p_3)$\\
\hline
$\nu_i+\bar\nu_j\rightarrow\nu_i+\bar\nu_j$ & $32G_F^2(p_1\cdot p_4)(p_2\cdot p_3)$\\
\hline
$\nu_i+\nu_j\rightarrow\nu_i+\nu_j$&$32G_F^2(p_1\cdot p_2)(p_3\cdot p_4)$\\
\hline
$\nu_i+\bar\nu_i\rightarrow e^++e^-$ & $128G_F^2[\tilde{g}_L^2(p_1\cdot p_4)(p_2\cdot p_3)+g_R^2(p_1\cdot p_3)(p_2\cdot p_4)+\tilde{g}_Lg_Rm_e^2(p_1\cdot p_2)]$\\
\hline
$\nu_i+e^-\rightarrow\nu_i+e^-$ & $128G_F^2[\tilde{g}_L^2(p_1\cdot p_2)(p_3\cdot p_4)+g_R^2(p_1\cdot p_4)(p_2\cdot p_3)-\tilde{g}_Lg_Rm_e^2(p_1\cdot p_3)]$\\
\hline
$\nu_i+e^+\rightarrow\nu_i+e^+$ & $128G_F^2[g_R^2(p_1\cdot p_2)(p_3\cdot p_4)+\tilde{g}_L^2(p_1\cdot p_4)(p_2\cdot p_3)-\tilde{g}_Lg_Rm_e^2(p_1\cdot p_3)]$\\
\hline
\end{tabular}
\caption{Matrix elements~\cite{Dolgov_Hansen} for $\mu$ and $\tau$ neutrino processes where $i=\mu,\tau$, $j=e,\mu,\tau$, $j\neq i$,  $\tilde{g}_L=g_L-1=-\frac{1}{2}+\sin^2\theta_W$, $g_R=\sin^2\theta_W$.}
\label{table:nu_mu_reac}
\end{table}

\subsection{Weinberg Angle}\label{ssec:matrix}
The $SU(2)\times U(1)$ gauge coupling constants $g$, $\gp$ are constrained by the two physical parameters, the Weinberg angle $\theta_W$ and the electric charge $e$
\begin{equation}\label{eq:gaugeconstr}
\sin \theta_W =\frac{\gp}{\sqrt{g^2+\gpt}}, \qquad e=\frac{g\gp}{\sqrt{g^2+\gpt}}.
\end{equation}
An alternative way to write these constraints is 
\begin{equation}\label{eq:gaugeconstr2}
g\sin\theta_W =e,\quad \gp\cos\theta_W =e,\qquad \frac 1{e^2}=\frac 1{g^2}+\frac 1{\gpt}.
\end{equation}
$\theta_W$  enters the  matrix elements presented in tables \ref{table:nu_e_reac} and \ref{table:nu_mu_reac} by way of 
\begin{equation}
g_L=\frac{1}{2}+\sin^2\theta_W,\quad g_R=\sin^2\theta_W,\qquad 
\tilde{g}_L=g_L-1=-\frac{1}{2}+\sin^2\theta_W.
\end{equation}

The Fermi constant $G_F$ fixes the vacuum expectation value of the Higgs field 
\begin{equation}
v\equiv 2^{-1/4}G_F^{-1/2}=246.22\,{\rm GeV}. 
\end{equation}
The mass of the $W$ and $Z$ gauge bosons can be written in terms of  $v$ 
\begin{equation}\label{eq:massWZ}
M_W=\frac{v}{2}g =\frac{ve}{2\sin \theta_W}, \qquad
M_Z=\frac{v}{2}\sqrt{g^2+\gpt} =\frac{ve}{2\sin \theta_W\cos \theta_W}.
\end{equation}
We show the dependence of $M_W$ and $M_Z$ on Weinberg angle in figure \ref{fig:W_Z_masses}, using $ve/2=38.4\GeV$ at the Z-scale.

\begin{figure}
\centerline{\includegraphics[width=0.75\columnwidth]{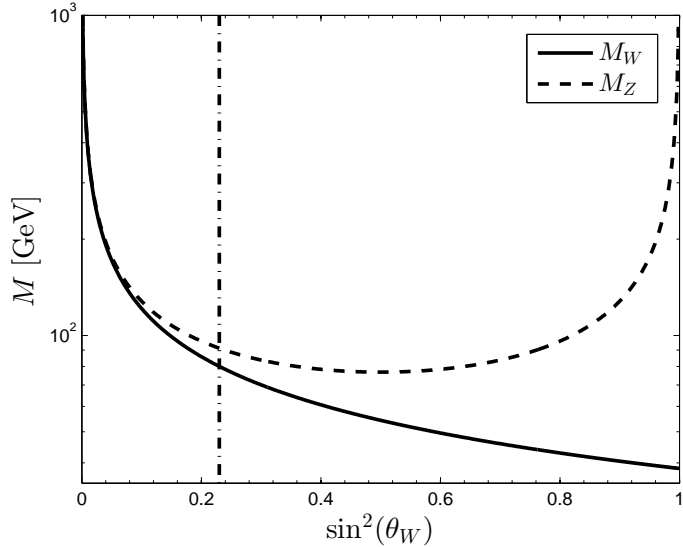}}
\caption{Dependence of $W$ and $Z$ boson masses on Weinberg angle, using $ve/2=38.4\GeV$ at the Z-scale. The vertical line refers to the SM vacuum case.}\label{fig:W_Z_masses}
 \end{figure}

Fixing $v$ and  $e$, $M_W$ is minimized when $\sin \theta_W =1$. Thus we find a factor of $1/2$ as a maximal possible reduction in $M_W$, also seen in figure~\ref{fig:W_Z_masses}.  This implies that $M_Z>M_W\gg |p|$ for neutrino momentum $p$ in the energy range of neutrino freeze-out, around $1\MeV$, even as we vary $\sin \theta_W $.  Even if $v$ is allowed to vary, for this approximation to cease  to be valid it would have to be reduced by a factor of $10^5$, in our view an extreme amount. Therefore we can carry out the computation of neutrino decoupling within the effective Fermi theory of weak interactions.

The ratio  ${M_W}/{M_Z}=\cos \theta_W =0.8815$ implies  $\sin^2\theta_W=0.223$. Considering that there is a rapid change  with scale the actual values are  $\left.\sin^2\theta_W\right|_{M_Z}=0.23116\pm0.00012$ and  $\left.\sin^2\theta_W\right|_{M_W}=0.22296\pm0.00028$. We will present our results as a function of $\sin^2\theta_W$ which we consider to be an unknown parameter in the hot Universe aged about one second. The other SM parameter of the electro-weak theory is the electric charge $e$. A variation in  $e$ is also possible, for example, due to  time evolution of the grand unified scale~\cite{Calmet:2002jz}. 

The symmetry breaking parameter $\sin^2\theta_W$ is at present a measured but theoretically unconstrained  SM parameter. However, should a grand unified approach in which the strong interactions are merged into the electroweak interactions be discovered, then presumably $\sin^2\theta_W$ could become fixed by the particular group structure. Such models are strongly constrained by  proton decay limits~\cite{Babu:2013jba}, hence a fundamental constraint on  $\sin^2\theta_W$  is not (yet) in sight.

\subsection{Interaction Strength Parameter $\eta$}\label{ssec:int_st}
In order to isolate the combination of natural constants which controls the neutrino freeze-out process, we cast the Einstein-Boltzmann model of neutrino freeze-out into dimensionless form. In the first step we look at the expansion of the Universe i.e. the Hubble parameter $H$. The Einstein equations contain the Hubble equation
\begin{equation}\label{hubble_eq}
(M_pH)^2=\frac{\rho}{3},\quad M_p=2.4354\,10^{18}\,{\rm GeV}, 
\end{equation}
where $\rho=T_{0}^{0}$ is the total gravitating energy density of the Universe and, as is usual in the context of general relativity, the Planck mass $M_p$ incorporates the factor $8\pi$ in the definition, \req{eq:etaMp}.

The divergence freedom of the Einstein equations requires  divergence freedom of the stress energy tensor $T^{\mu\nu}$, a condition which  reads for a homogeneous Universe
\begin{equation}\label{div_T}
\frac{\dot\rho}{\rho+P}=-3H .
\end{equation}
Combining \req{hubble_eq} with \req{div_T} shows that time change occurs at scale $\tau\propto M_p/\sqrt{\rho}$. In the domain of interest the energy density $\rho$ is  characterized  by the electron mass  $\rho \propto m_e^4$. The scale $m_e$ is related both to the key energy component of the Universe at the time of neutrino freeze-out and the ambient temperature. We  thus recognize the time scale to be characterized by  $\tau\equiv M_p/m_e^2=6.12$\,s; the actual time scale is close to 1s considering the presence of many degrees of freedom.

Using the timescale $\tau$, and scaling all momenta, energies, energy densities, pressures, and temperatures by the appropriate power of $m_e$ we can combine  all scale dependent parameters in the Einstein-Boltzmann equation. We thus find  
\begin{equation}\label{eta_bol}
\partial_tf-pH\partial_pf=\eta \frac{C[f]}{E},
\end{equation}
where in the interaction strength $\eta$, \req{eq:etaMp}, we include  the $G_F^2$ factor  common to all of the neutrino interaction matrix elements.

Aside from the $\theta_W$ dependence of the matrix elements seen in tables \ref{table:nu_e_reac} and \ref{table:nu_mu_reac}, the complete dependence on natural constants is now contained in a single dimensionless interaction strength parameter $\eta$ with the vacuum present day value,
\begin{equation}\label{eta_def}
\eta_0\equiv \left.M_p m_e^3 G_F^2\right|_0  = 0.04421 .
\end{equation}
If the dominant component of the electron mass originates in the Higgs mechanism, we find somewhat different scaling   $\eta  \propto g_{Ye}^3M_p/v$, where Yukawa electron coupling is introduced $ g_{Ye}\simeq v/m_e$ .

The discussion we presented is only focused on the normalization by natural constants of  the collision term. The magnitude of the scattering integrals also depends on the magnitude of the scaled temperature $T/m_e$.  In particular, in the limit $T/m_e<1$ the scattering integrals involving $e^\pm$ neutrino scattering are suppressed exponentially by a factor $e^{-m_e/T}$ or $e^{-2m_e/T}$ due to the diminished presence of $e^\pm$ pairs. However, our objective in writing \req{eta_bol} was not to isolate the leading order behavior, but rather to separate out all dependence on dimensioned natural constants and isolated them in the interaction strength parameter $\eta$.  This means that, as a dynamical system, the solutions of the dimensionless form \req{eta_bol} depends only on the parameters $\eta$ and $\sin^2\theta_W$, and hence all quantities computed from solutions of the Boltzmann equation that are dimensionless, such as $N_\nu$, can also only depend on $\eta$ and  $\sin^2\theta_W$.  Of course dimensioned quantities, for example the magnitude of freeze-out temperatures, still have to be scaled appropriately i.e. energies must be multiplied by $m_e$ and times must be multiplied by the timescale $\tau$, and so dimensioned quantities will show an additional dependence on natural constants.

Our argument that there are only two dimensionless variables of interest, $\eta, \sin^2\theta_W$, relies on the fact that there is only one particle scale parameter that enters the energy density and collision integrals, namely $m_e$. This is so since for $T\in (0.1,3)$\,MeV, muons are too heavy, $m_\mu=105.66$ MeV, the baryon energy density  controlled by $M_B-\mu_B$ is too small, and all other energy components in Universe are completely negligible. Thus though in principle $N_\nu=N_\nu(\eta,\sin^2\theta_W, m_e/m_\mu, m_e/(M_B-\mu_B))$, we can safely ignore  all additional dimensionless quantities. Furthermore, given our hypothesis that a modification of SM parameters in the early Universe could contribute to $N_\nu$, there is also a contribution to the Universe dynamics from the rate of change of these parameters. We assume that any such rate of change is small enough to be insignificant and will not discuss it further.

The dependence of  $N_\nu(\eta,\sin^2 \theta_W)$  will be the key result of this work and is presented below in the section \ref{sec:param_dep}. Qualitatively, it is apparent that an increase in $N_\nu$  requires increased coupling strength $\eta$. The dependence on $\sin^2\theta_W$  is much less obvious in view of the gauge boson  mass $M_W,M_Z $ variation, see figure \ref{fig:W_Z_masses}. The key  question we aim to resolve in this work is how sensitive  is $N_\nu$  to  a change in $\eta$ and $\sin^2\theta_W$. 

\section{Solving the Relativistic Boltzmann Equation}\label{sec:math}
\subsection{Emerging Chemical Nonequilibrium Method}\label{sec:boltz_solve}
We solve the Boltzmann equation \req{boltzmann_p} by the spectral method detailed in \cite{Birrell_orthopoly}.  We give only a brief outline of the method here.  Our approach is adapted to systems near kinetic equilibrium (i.e. equilibrium momentum distribution) but not necessarily chemical equilibrium (i.e. allowing for non-equilibrium particle number yield), allowing for potentially large reheating.  In other words, the method performs best when the distribution is of the form
\begin{equation}
f(t,p)=f_\Upsilon(t,p)(1+\phi(t,p)),\hspace{2mm} f_\Upsilon(t,p)=\frac{1}{\Upsilon^{-1} e^{p/T}+1}
\end{equation}
where $\phi$ is small and $T$ and $\Upsilon$ are the dynamical effective temperature and fugacity (i.e. phase space occupation parameter) respectively. Since we adapt both $T$ and $\Upsilon$ as function of time, we employ  a moving (in Hilbert space) frame, in  which  the orthogonal polynomial basis  dynamically evolves to suit the problem. 

Our approach should be contrasted with the method used in \cite{Esposito2000,Mangano2002}, which we call the chemical equilibrium method, that studied neutrino freeze-out using a fixed orthogonal polynomial basis generated by the chemical equilibrium weight
\begin{equation}
f_c=\frac{1}{e^y+1},\hspace{2mm} y=a(t)p.
\end{equation}
We note in the above that the temperature scaling is also assumed, that is $Ta(t)=$Const. In our approach we allow for reheating of the effective temperature to occur and thus also $T$, like  $\Upsilon$, evolves in time independently.

The deviation from chemical equilibrium is characterized in $f_\Upsilon$ by the fugacity $\Upsilon$.  A non-equilibrium $\Upsilon\neq 1$ builds up during neutrino freeze-out, specifically in the temperature range where the process $e^+e^-\rightarrow \nu\bar\nu$ is too slow to equilibrate particle number but $e^\pm \nu\rightarrow e^\pm \nu$ scattering is still able to equilibrate momentum.  The introduction of chemical non-equilibrium through $\Upsilon\ne 1$  contrasts with the chemical equilibrium method described above.  

The chemical equilibrium method is appropriate for the physical regime studied  in \cite{Esposito2000,Mangano2002}, wherein neutrinos are almost entirely decoupled by the time of $e^\pm$ annihilation and therefore there is little time for reheating of neutrinos or the development of chemical non-equilibrium.  However, for our purposes, namely the characterization of $N_\nu(\eta,\sin^2\theta_W)$, we must use a method that does not rely on small coupling for its effectiveness, hence we were motivated to develop the method described here. A comparison with the results of the chemical equilibrium method is found in section \ref{sec:comp}. We refer to Ref.\cite{Birrell_orthopoly} for further discussion and detailed validation of the method we present. 

After changing variables $z=p/T$, we will solve \req{boltzmann_p} by expanding $\phi$ in the basis of orthonormal polynomials, $\hat\psi_i(z)$, generated by the parametrized weight function
\begin{equation}\label{weight}
w(z)\equiv w_\Upsilon(z)\equiv \frac{z^2}{\Upsilon^{-1} e^z+1}
\end{equation}
on the interval $[0,\infty)$
\begin{equation}
\phi(t,z)=\sum_{i=0}^\infty b^i(t)\hat\psi_i(z).
\end{equation}
By convention, they are indexed so that $\hat\psi_j$ has degree $j$. This choice of the weight is physically motivated by the phase space  of practically massless neutrinos emerging into a chemical non-equilibrium distribution. 

The Boltzmann equation then results in an equation for the mode coefficients~\cite{Birrell_orthopoly}
\begin{align}\label{b_eq}
\dot b^k=& \left(H+\frac{\dot{T}}{T}\right)\sum_i A_i^k(\Upsilon)b^i-\frac{\dot{\Upsilon}}{\Upsilon}\sum_i B_i^k(\Upsilon)b^i+\langle\frac{1}{f_\Upsilon E}C[f],\hat\psi_k\rangle
\end{align}
where the matrices  $A$ and $B$ are
\begin{align}\label{A_B_matrices}
A^k_i(\Upsilon)\equiv&\langle\frac{z}{f_k}\hat\psi_i\partial_zf_k,\hat\psi_k\rangle+\langle z\partial_z \hat\psi_i,\hat\psi_k\rangle=\langle\frac{-z}{1+\Upsilon e^{-z}}\hat\psi_i,\hat\psi_k\rangle+\langle z\partial_z \hat\psi_i,\hat\psi_k\rangle,\\
B^k_i(\Upsilon)\equiv &\Upsilon\left(\langle\frac{1}{f_k}\frac{\partial f_k}{\partial\Upsilon}\hat\psi_i,\hat\psi_k\rangle+\langle\frac{\partial\hat{\psi}_i}{\partial \Upsilon},\hat\psi_k\rangle\right)=\langle\frac{1}{1+\Upsilon e^{-z}}\hat\psi_i,\hat\psi_k\rangle+\Upsilon\langle\frac{\partial\hat{\psi}_i}{\partial \Upsilon},\hat\psi_k\rangle.
\end{align}
For details on how to construct the inner products $\langle\frac{\partial\hat{\psi}_i}{\partial \Upsilon},\hat\psi_k\rangle$ we refer to Appendix A of Ref.\cite{Birrell_orthopoly}. 

The dynamics of the effective temperature and fugacity are fixed by the requirement that $f_\Upsilon$ captures the number density and energy density of the full distribution $f$, leaving $\phi$ to describe only the non-thermal distortions.   In practice, this implies that $b^0(t)=b^1(t)=0$ and a minimum of only two degrees of freedom (or modes), $T$ and $\Upsilon$, are required for our method.  See Ref.~\cite{Birrell_orthopoly} for details on the resulting evolution equations for $T(t)$ and $\Upsilon(t)$.  In contrast, we note that the minimum number modes required for the chemical equilibrium method is four.

\subsection{Collision Integral Inner Products}\label{sec:coll_simp}
In order to solve for the mode coefficients, the inner products of  collision integrals  with respect to the weight function \req{weight},
\begin{equation}
R_k\equiv\langle\frac{1}{f_\Upsilon E_1}C[f_1],\hat\psi_k\rangle,
\end{equation}
must be computed.  
\begin{align}\label{collision_integrals}
R_k
=&\int_0^\infty \hat\psi_k(z_1)C[f_1](z_1) \frac{z_1^2}{E_1}dz_1\\
=&\frac{1}{2}\int \hat\psi_k(z_1)\int\left[f_3(p_3)f_4(p^4)f^1(p_1)f^2(p_2)-f_1(p_1)f_2(p_2)f^3(p_3)f^4(p^4)\right]\\
&\hspace{20mm}\times S |\mathcal{M}|^2(s,t)(2\pi)^4\delta(\Delta p)\prod_{i=2}^4\frac{d^{3}p_i}{2(2\pi)^3E_i}\frac{z_1^2}{E_1}dz_1,\notag\\
=&\frac{2(2\pi)^3}{8\pi}T_1^{-3}\int G_k(p_1,p_2,p_3,p_4)S |\mathcal{M}|^2(s,t)(2\pi)^4\delta(\Delta p)\prod_{i=1}^4\frac{d^{3}p_i}{2(2\pi)^3E_i},\\
=&2\pi^2T_1^{-3}\int G_k(p_1,p_2,p_3,p_4)S |\mathcal{M}|^2(s,t)(2\pi)^4\delta(\Delta p)\prod_{i=1}^4 \delta_0(p_i^2-m_i^2)\frac{d^4p_i}{(2\pi)^3}\\
G_k=&\hat\psi_k(z_1)\left[f_3(p_3)f_4(p_4)f^1(p_1)f^2(p_2)-f_1(p_1)f_2(p_2)f^3(p_3)f^4(p_4)\right],\hspace{2mm} f^i=1- f_i.
\end{align}
The matrix element for a $2-2$ reaction can be written as a function of the Mandelstam variables $s, t, u$, of which only two are independent, defined by 
\begin{align}\label{Mandelstam}
&s=(p_1+p_2)^2=(p_3+p_4)^2,\\
&t=(p_3-p_1)^2=(p_2-p_4)^2,\\
&u=(p_3-p_2)^2=(p_1-p_4)^2,\\
&s+t+u=\sum_i m_i^2,
\end{align}
and we will consider this done for the analysis that follows.

Note that $R_k$ only uses information about the distributions at a single spacetime point, and so we can work in a local orthonormal basis for the momentum.  Among other things, this implies that $p^2=p^\alpha p^\beta\eta_{\alpha\beta}$ where $\eta$ is the Minkowski metric
\begin{equation}
\eta_{\alpha\beta}=\diag(1,-1,-1,-1).
\end{equation}

From \req{collision_integrals}, we see that a crucial prerequisite of our spectral method is the capability to evaluate  integrals of the type
\begin{align}\label{coll_ip}
&M\equiv\int G(p_1,p_2,p_3,p_4) S |\mathcal{M}|^2(s,t) (2\pi)^4\delta(\Delta p)\prod_{i=1}^4 \delta_0(p_i^2-m_i^2)\frac{d^4p_i}{(2\pi)^3},\\
&G(p_1,p_2,p_3,p_4)=g_1(p_1)g_2(p_2)g_3(p_3)g_4(p_4)
\end{align}
for some functions $g_i$. Even after eliminating the delta functions in \req{coll_ip}, we are still left with an $8$-dimensional integral.  To facilitate numerical computation, we analytically reduce this expression down to fewer dimensions.  Fortunately, the systems we are interested in have a large amount of symmetry that we can utilize for this purpose.  

The distribution functions we are concerned with are isotropic in some frame defined  by a unit timelike vector $U$, i.e. they depend   on the four-momentum $p_i$ only through the quantities $p_i\cdot U$  and $p_i^2=m_i^2$.  The same is true of the basis functions $\hat\psi_k$ and hence we can assume the  $g_i$ depend only on $p_i\cdot U$ as well.  In \cite{Madsen,Dolgov_Hansen} approaches are outlined that reduce integrals of this type down to $3$ dimensions.   However, the integrand one obtains from these methods is only piecewise smooth or has an integration domain with a complicated geometry.  This can present difficulties for numerical integration routines and so we take an alternative approach that, for the scattering kernels found in $e^\pm$, neutrino interactions, reduces the problem to three iterated integrals (but not quite to a three dimensional integral) and results in an integrand with better smoothness properties. Depending on the integration method used, this can significantly reduce the numerical cost of evaluating the collision integrals.  The derivation presented expands on what is found in Ref.\cite{letessier2002hadrons}.

\subsection{Simplifying the Collision Integral}\label{sec:coll_simp3}
Our strategy for simplifying the collision integrals is as follows.  We first make a change of variables designed to put the 4-momentum conserving delta function in a particularly simple form, which allows us to analytically use that delta function to reduce the integral from $16$ to $12$ dimensions.  The remaining four delta functions, which impose the mass shell constraints, are then seen to reduce to integration over a product of spheres.  The simple form of the submanifold that these delta function restrict the integration to allows the method described in \ref{app:delta} to analytically evaluate all four of the remaining delta functions simultaneously.  During this process, the isotropy of the system in the frame given by the 4-vector $U$ allows us to reduce the dimensionality further, by analytically evaluating several of the angular integrals. 

The change of variables that simplifies the 4-momentum conserving delta function is given by
\begin{equation}
p=p_1+p_2,\hspace{2mm} q=p_1-p_2, \hspace{2mm} \pp =p_3+p_4, \hspace{2mm} \qp =p_3-p_4.
\end{equation}
The Jacobian of this transformation is $1/2^{8}$.  Therefore, changing variables in the delta functions we find

\begin{align}\label{M_eq1}
M=
b\!\!\int &G\left((p+q)\cdot U/2,(p-q)\cdot  U/2,(\pp +\qp )\cdot U/2, (\pp -\qp )\cdot U/2\right)S |\mathcal{M}|^2 \notag\\ 
\times   &\Theta(p^0-|q^0|)\Theta\left((\pp )^0-|(\qp )^0|\right) \delta(p-\pp )\delta\left((p+q)^2/4-m_1^2\right)\delta\left((p-q)^2/4-m_2^2\right)\notag\\[0.1cm]
\times&\delta\left((\pp +\qp )^2/4-m_3^2\right)\delta\left((\pp -\qp )^2/4-m_4^2\right)d^4pd^4qd^4\pp d^4\qp  
\end{align}
where $\Theta(x)$ denotes the Heaviside function, $b= {1}/{256(2\pi)^8}$, and $U$ is the four velocity characterizing the isotropic frame as discussed above.

Using the coarea formula, theorem \ref{vol_form_coarea} in \ref{app:delta}, we decompose this into an integral over $s=p^2$, the center of mass energy, and also eliminate the integration over $\pp $ using $\delta(p-\pp )$,
\small
\begin{align}\label{M_eq2}
M= b\!\!&\int_{s_0}^\infty \!\!\int \delta(p^2-s)\! \left[\int S |\mathcal{M}|^2  F(p,q,\qp )  \Theta(p^0-|q^0|)\Theta\left((\pp )^0-|(\qp )^0|\right) \delta\left((p+q)^2/4-m_1^2\right)\right.\notag\\
\times&\delta\left((p-q)^2/4-m_2^2\right)\delta\left((p+\qp )^2/4-m_3^2\right)\delta\left((p-\qp )^2/4-m_4^2\right)d^4qd^4\qp \bigg]d^4pds,\\[0.1cm]
F(p,&q,\qp )=G\left((p+q)\cdot U/2,(p-q)\cdot  U/2,(p+\qp )\cdot U/2, (p-\qp )\cdot U/2\right),\notag\\[0.1cm]
s_0=&\max\{(m_1+m_2)^2,(m_3+m_4)^2\}.\notag
\end{align}
\normalsize
The lower bound on $s$ comes from the fact that both $p_1$ and $p_2$ are future timelike and hence 
\begin{equation}
p^2=m_1^2+m_2^2+2p_1\cdot p_2\geq m_1^2+m_2^2+2m_1m_2=(m_1+m_2)^2.
\end{equation}
The other inequality is obtained using $p=\pp $. 

Note that the integral in brackets in \req{M_eq2} is invariant under $SO(3)$ rotations of $p$ in the frame defined by $U$.  Therefore we obtain
\begin{align}\label{K_def}
M=&  b\!\!\int_{s_0}^\infty\!\!\int_0^\infty K(s,p)\frac{4\pi |\vec p|^2}{2p^0}d|\vec p|ds,\hspace{2mm} p^0=p\cdot U=\sqrt{|\vec p|^2+s},\\
K(s,p)=&\!\int\!\! S |\mathcal{M}|^2  F(p,q,\qp )  \Theta(p^0-|q^0|)\Theta\left((\pp )^0-|(\qp )^0|\right)  \delta\left((p+q)^2/4-m_1^2\right)\,\notag\\[0.1cm]
\times&\delta\!\left((p-q)^2/4-m_2^2\right)\delta\left((p+\qp )^2/4-m_3^2\right)\delta\left((p-\qp )^2/4-m_4^2\right)d^4qd^4\qp 
\end{align}
where $|\vec p|$ denotes the norm of the spacial component of $p$ and in the formula for $K(s,p)$, $p$ is any four vector whose spacial component has norm $|\vec p|$ and timelike component $\sqrt{|\vec p|^2+s}$. Note that in integrating over $\delta(p^2-s)dp^0$, only the positive root was taken, due to the Heaviside functions in the $K(s,p)$.

We now simplify $K(s,p)$ for fixed but arbitrary $p$ and $s$ that satisfy $p^0=\sqrt{|\vec p|^2+s}$ and $s>s_0$.  These conditions imply $p$ is future timelike, hence we can we can change variables in  $q,\qp $ by an element of $Q\in SO(1,3)$ so that 
\begin{equation}\label{Eq:U1}
Qp=(\sqrt{s},0,0,0), \hspace{2mm} QU=(\alpha,0,0,\delta)
\end{equation}
where
\begin{equation}\label{Eq:U2}
\alpha=\frac{p\cdot U}{\sqrt{s}}, \hspace{2mm} \delta=\frac{1}{\sqrt{s}}\left((p\cdot U)^2-s \right)^{1/2}.
\end{equation}
Note that the delta functions in the integrand imply $p\pm q$ is  timelike (or null if the corresponding mass is zero).  Therefore $p^0>\pm q^0$ iff $p\mp q$ is future timelike (or null).  This condition is preserved by $SO(1,3)$, hence $p^0>|q^0|$ in one frame iff it holds in every frame.  Similar comments apply to $p^0>|(\qp )^0|$ and so $K(s,p)$ has the same formula in the transformed frame as well.

We now evaluate the measure that is induced by the delta functions, using the method given in \ref{app:delta}.  We have the constraint function
\begin{equation}
\Phi(q,\qp )=\left((p+q)^2/4-m_1^2,(p-q)^2/4-m_2^2,(p+\qp )^2/4-m_3^2,(p-\qp )^2/4-m_4^2\right)
\end{equation}
and must compute the solution set $\Phi(q,\qp )=0$. Adding and subtracting the first two components and the last two respectively, we have the equivalent conditions
\begin{align}
\frac{s+q^2}{2}=m_1^2+m_2^2,\hspace{2mm} p\cdot q=m_1^2-m_2^2, \hspace{2mm}\frac{s+(\qp )^2}{2}=m_3^2+m_4^2,\hspace{2mm} p\cdot \qp =m_3^2-m_4^2.
\end{align}
If we let $(q^0,\vec{q})$, $((\qp )^0,\vec q\,')$ denote the spacial components in the frame defined by $p=(\sqrt{s},0,0,0)$ we have the equivalent conditions
\begin{align}\label{coord_conditions}
&q^{0}=\frac{m_1^2-m_2^2}{\sqrt{s}},\hspace{2mm} |\vec{q}|^2=\frac{(m_1^2-m_2^2)^2}{s}+s-2(m_1^2+m_2^2),\\
&(\qp )^{0}=\frac{m_3^2-m_4^2}{\sqrt{s}}, \hspace{2mm} |\vec q\,'|^2=\frac{(m_3^2-m_4^2)^2}{s}+s-2(m_3^2+m_4^2).
\end{align}
Note that the above formulas, together with $s\geq s_0$, imply
\begin{equation}
\frac{|q^0|}{p^0}\leq \frac{|m_1^2-m_2^2|}{(m_1+m_2)^2}<1
\end{equation}
and similarly for $\qp $.  Hence the Heaviside functions are identically equal to $1$ under these conditions and we can drop them from the formula for $K(s,p)$.

The conditions \req{coord_conditions} imply that our solution set is a product of spheres in $\vec{q}$ and $\vec q\,'$, as long as the conditions are consistent i.e. so long as $|\vec{q}|,|\vec q\,'|>0$. To see that this holds for almost every $s$, first note
\begin{equation}
\frac{d}{ds}|\vec{q}|^2=1-\frac{(m_1^2-m_2^2)^2}{s^2}>0
\end{equation}
since $s\geq (m_1+m_2)^2$.  At $s=(m_1+m_2)^2$, $|\vec{q}|^2=0$.  Therefore, for $s>s_0$ we have $|\vec{q}|>0$ and similarly for $\qp $.  Hence we have the result
\begin{equation}
\Phi^{-1}(0)=\{q^{0}\}\times B_{|\vec{q}|}\times \{(\qp )^{0}\}\times B_{|\vec q\,'|}.
\end{equation}
where $B_r$ denotes the radius $r$ ball centered at $0$.  We will parametrize this by spherical angular coordinates in $q$ and $\qp $. 

 We now compute the induced volume form.  First consider the differential 
\begin{equation}
 D\Phi=\left( \begin{array}{c}
\frac{1}{2}(q+p)^\alpha\eta_{\alpha\beta}dq^\beta \\
\frac{1}{2}(q-p)^\alpha\eta_{\alpha\beta}dq^\beta\\
\frac{1}{2}(\qp +p)^\alpha\eta_{\alpha\beta}dq^{'^\beta}  \\
\frac{1}{2}(\qp -p)^\alpha\eta_{\alpha\beta}dq^{'^\beta}  \end{array} \right).
\end{equation}
Evaluating this on the coordinate vector fields $\partial_{q^0}$, $\partial_r$ we obtain
\begin{equation}
 D\Phi(\partial_{q^0})=\left( \begin{array}{c}
\frac{1}{2}(q^0+\sqrt{s}) \\
\frac{1}{2}(q^0-\sqrt{s}) \\
0\\
0 \end{array} \right), \hspace{2mm}  D\Phi(\partial_{r})=\left( \begin{array}{c}
-\frac{1}{2}|\vec{q}| \\
-\frac{1}{2}|\vec{q}| \\
0\\
0 \end{array} \right)=\left( \begin{array}{c}
-\frac{1}{2}r \\
-\frac{1}{2}r \\
0\\
0 \end{array} \right).
\end{equation}
Similar results hold for $\qp $.  Therefore we have the determinant
\begin{equation}
\det\left( \begin{array}{cccc}
D\Phi(\partial_{q^0}) & D\Phi(\partial_{r}) & D\Phi(\partial_{(\qp )^0}) & D\Phi(\partial_{\rp }) \end{array} \right)=\frac{s}{4}r\rp .
\end{equation}
By corollary \ref{induced_vol_eq} and \req{delta_def} in \ref{app:delta}, this implies that the induced volume measure is
\begin{align}
&\delta\left((p+q)^2/4-m_1^2\right)\delta\left((p-q)^2/4-m_2^2\right)\delta\left((p+\qp )^2/4-m_3^2\right)\delta\left((p-\qp )^2/4-m_4^2\right)d^4qd^4\qp \notag\\
=&\frac{4}{sr\rp }i_{(\partial_{q^0},\partial_{r},\partial _{(\qp )^0},\partial_{\rp })}\left[\left(r^2\sin(\phi)dq^0drd\theta d\phi\right)\wedge\left((\rp )^2\sin(\phip )d(\qp )^0d\rp d\thetap d\phip \right)\right]\notag \\
=&\frac{4r\rp }{s}\sin(\phi)\sin(\phip )d\theta d\phi d\thetap d\phip
\end{align}
where
\begin{align}
r=&\frac{1}{\sqrt{s}}\sqrt{(s-(m_1+m_2)^2)(s-(m_1-m_2)^2)},\notag\\ \rp =&\frac{1}{\sqrt{s}}\sqrt{(s-(m_3+m_4)^2)(s-(m_3-m_4)^2)}
\end{align}
and $i$ is the interior product (i.e. contraction) operator as described in \ref{app:delta}.
 
Consistent with our interest in the Boltzmann equation, we assume $F$ factors as
\begin{align}
 F(p,q,\qp )=&F_{12}\left((p+q)\cdot U/2,(p-q)\cdot U/2)F_{34}((p+\qp )\cdot U/2,(p-\qp )\cdot U/2\right)\\
\equiv &G_{12}(p\cdot U,q\cdot U)G_{34}(p\cdot U,\qp \cdot U).
\end{align}
For now, we suppress the dependence on $p$, as it is not of immediate concern. In our chosen coordinates where $U=(\alpha,0,0,\delta)$ we have
\begin{equation}
q\cdot U=q^0\alpha-r\delta\cos(\phi)
\end{equation}
and similarly for $\qp $.

To compute
\begin{align}\label{K_angular1}
K(s,p)=\frac{4r\rp }{s}\int \left[\int S |\mathcal{M}|^2 (s,t) G_{34}\sin(\phip )d\thetap  d\phip \right] G_{12}\sin(\phi)d\theta d\phi
\end{align}
first recall 
\begin{align}
t=&(p_1-p_3)^2=\frac{1}{4}(q- \qp )^2=\frac{1}{4}\left(q^2+(\qp )^2-2(q^0(\qp )^0-\vec{q}\cdot \vec q\,')\right),\\
\vec{q}\cdot\vec q\,'&=r\rp \left(\cos(\theta-\thetap )\sin(\phi)\sin(\phip )+\cos(\phi)\cos(\phip )\right).
\end{align}
Together, these imply that the integral in brackets in  \req{K_angular1} equals
\begin{align}
&\int_0^\pi\int_0^{2\pi} S |\mathcal{M}|^2 (s,t(\cos(\theta-\thetap )\sin(\phi)\sin(\phip )+\cos(\phi)\cos(\phip )))\\
&\hspace{15mm}\times G_{34}\left((\qp )^0\alpha-\rp \delta\cos(\phip )\right)\sin(\phip )d\thetap  d\phip \notag\\
=&\int_{-1}^1\int_0^{2\pi} S |\mathcal{M}|^2 (s,t(\cos(\psi)\sin(\phi)\sqrt{1-y^2}+\cos(\phi)y)) G_{34}\left((\qp )^0\alpha-\rp \delta y\right)d\psi dy.
\end{align}

Therefore 
\begin{align}\label{M_simp}
K(s,p)=&\frac{8\pi r\rp }{s}\int_{-1}^1 \left[\int_{-1}^1\left(\int_0^{2\pi} S |\mathcal{M}|^2\left (s,t(\cos(\psi)\sqrt{1-y^2}\sqrt{1-z^2}+yz)\right)d\psi\right)\right.\\
&\hspace{26mm}\times G_{34}\left((\qp )^0\alpha-\rp \delta y\right) dy\bigg] G_{12}(q^0\alpha-r\delta z)dz\notag
\end{align}
where
\begin{align}\label{t_def}
t(x)=&\frac{1}{4}\left((q^0)^2-r^2+((\qp )^0)^2-(\rp )^2-2q^0(\qp )^0+2r\rp x\right),\\
=&\frac{1}{4}\left((q^0-(\qp )^0)^2-r^2-(\rp )^2+2r\rp x\right).
\end{align}
This is as far as we can simplify things without more information about the form of the matrix elements. In \ref{app:nu_matrix_elements} we apply this method and analytically simplify \req{M_simp} for each of the processes in tables \ref{table:nu_e_reac} and \ref{table:nu_mu_reac} as much as possible and in the process we show that $M$ can be written in terms of three iterated integrals for each of these processes.

\subsection{Validation}\label{sec:comp}
We solve the Boltzmann equation, \req{boltzmann}, for both the electron neutrino distribution and the combined $\mu$, $\tau$ neutrino distribution, including all of the  processes from tables \ref{table:nu_e_reac} and \ref{table:nu_mu_reac} in the scattering operator, together with the Hubble equation for $a(t)$, \req{hubble_eq}.  The total energy density  appearing in the Hubble equation consists of the contributions from both neutrino distributions as well as chemical equilibrium $e^\pm$ and photon distributions at some common temperature $T_\gamma$.  The dynamics of $T_\gamma$ are fixed by the divergence freedom of the total stress energy tensor, \req{div_T}.  In addition, we include the QED corrections to the $e^\pm$ and photon equations of state as described in \cite{Mangano2002}.

We compared the results of numerically evaluating the collision integrals using our method as given in sections \ref{sec:coll_simp} and \ref{app:nu_matrix_elements} with the method used by Ref.\cite{Dolgov_Hansen} and validated that results agree, up to numerical integration tolerance. To compare our results from solving the Boltzmann equation with Ref.~\cite{Mangano2005}, where neutrino freeze-out was simulated using $\sin^2\theta_W=0.23$ and $\eta=\eta_0$, in table \ref{table:method_comp} we present $N_\nu$ together with the following quantities
\begin{equation}\label{eq:comp}
z_{\rm fin}=T_\gamma a,\qquad \rho_{\nu 0}=\frac{7}{120}\pi^2T^{4}, 
\qquad \delta\bar\rho_{\nu}= \frac{\rho_\nu}{\rho_{\nu 0}}-1.
\end{equation}
\begin{table}[h] 
\centering
\begin{tabular}{|c|c|c|c|c|c|}
\hline
Method &Modes&$z_{\rm fin}$ & $\delta\bar\rho_{\nu_e}$& $\delta\bar\rho_{\nu_{\mu,\tau}}$ & $N_{\nu}$ \\
\hline
Chemical Eq& 4 &1.39785 &0.009230 &0.003792 &3.044269\\
\hline
Chemical Non-Eq& 2&1.39784 &0.009269 & 0.003799&3.044383 \\
\hline
Chemical Non-Eq& 3&1.39785&0.009230 & 0.003791&3.044264 \\
\hline
\end{tabular}
\caption{Comparison of neutrino freeze-out results obtained in Ref.~\cite{Mangano2005} (top line) with those obtained using the methods described above, which allow for a reduced number of expansion modes.}
\label{table:method_comp}
\end{table}

The quantities presented in \req{eq:comp} and table \ref{table:method_comp}  were introduced in Ref.~\cite{Mangano2005}, but some additional discussion is in order. 
\begin{enumerate}
\item
The quantity $z_{\rm fin}$ measures the deviation of the photon temperature $T_\gamma$ from  the `free streaming' temperature $T\propto 1/a$, i.e the temperature  of a (hypothetical) particle species that is completely decoupled throughout the domain of temperature  considered. Therefore, $z_{\rm fin}=T_\gamma/T$ is the measure of the amount of reheating  photons underwent due to the annihilation of $e^\pm$. For the case of already completely decoupled neutrinos, whose temperature is in this case just the free-streaming temperature, according to \req{T_nu_T_gamma}
\begin{equation}
\left.z_{\rm fin}\right|_{\nu \ \rm decoupled}=(11/4)^{1/3}\approx 1.401.
\end{equation}
For the case of some $e^\pm$ annihilation occurring while neutrinos are still coupled, one expects this value to be slightly reduced, due to the transfer of some $e^\pm$ entropy  into neutrinos. This is reflected in the values seen in table \ref{table:method_comp}.
\item
$\rho_{\nu0}$ is the energy density of a single massless fermion with two degrees of freedom and temperature equal to the free-streaming temperature. In other words, it is the energy density of a single neutrino species, assuming it decoupled before reheating. Consequently, $\delta\bar\rho_\nu$ is the fractional increase in the energy density of a coupled neutrino species, due to its partial participation in reheating.
\end{enumerate}

The top entry in table \ref{table:method_comp} correspond to the reference values from Ref.~\cite{Mangano2005}. The next two lines present our results that use the chemical non-equilibrium method which, as we show, allows for a smaller basis set. We show 2 and 3 modes, respectively which compare  with  the 4 modes case for the equilibrium method. The value of $N_\nu$ we obtain agrees for the case of 2 modes with that found by \cite{Mangano2005}, up to their cited error tolerance. 

Considering both the improved smoothness properties of integrands developed in sections \ref{sec:coll_simp} and \ref{app:nu_matrix_elements} and  the reduction in the required number of modes for  the chemical non-equilibrium method, our approach  with the minimum number $2$  of required modes  is found to be more than $20\times$ faster than the chemical equilibrium method with its minimum number $4$ of required modes. This computational performance improvement makes it possible to explore the neutrino freeze-out process for many different circumstances and parameter sets employing a desktop PC.

\section{Dependence of Neutrino Freeze-out on Standard Model Parameters}\label{sec:param_dep}
\subsection{Neutrino Freeze-out Temperature}\label{sec:process_types}
SM parameters impact $N_\nu$ by changing how long the neutrinos remain coupled to the annihilating $e^\pm$ and thereby impacting the amount of energy and entropy transfer.  In other words, the neutrino freeze-out temperature is modified.  Before we present the dependence of $N_\nu$ on $\theta_W$ and $\eta$ we first consider in detail the freeze-out temperatures of neutrinos with an initial focus on the conventional SM parameters. 

In the literature one finds estimates of freeze-out temperatures based on a comparison of Hubble expansion with neutrino scattering length  and considering only number changing (i.e. chemical) processes, see e.g. Ref.\cite{kolb}. We employ a similar definition of freeze-out temperature in the context of the Boltzmann equation and refine the results by noting that there are three different freeze-out processes:
\begin{enumerate}
\item Neutrino chemical freeze-out: the neutrino pair number changing annihilation process 
\begin{equation}
l+\bar l\Leftrightarrow \nu_l+\bar \nu_l
\end{equation}
which we will see decouples at the highest temperature.
\item Neutrino kinetic freeze-out: The sharing of energy between leptons and neutrinos by way of scattering
\begin{equation}
l+\nu\Leftrightarrow  l+  \nu 
\end{equation}
stops at a lower energy compared to neutrino number changing processes.
\item  Collisions between neutrinos are capable of re-equilibrating energy within and between flavor families. These processes end at a yet lower temperature and the neutrinos will be truly free-streaming from that point on. 
\end{enumerate}

The attentive reader will notice that we have omitted here a discussion of flavor neutrino oscillations. If it weren't for the differences between the matrix elements for the interactions between $e^\pm$ and $\nu_e$ on one hand and $e^\pm$ and $\nu_\mu,\nu_\tau$ on the other, oscillations would have no effect on the flow of entropy into neutrinos and hence no effect on $N_\nu$.  However, there are differences and they do lead to a modification of $N_\nu$.  In Ref.~\cite{Mangano2005} the impact of oscillations on neutrino freeze-out for the present day measured values of $\theta_W$ and $\eta$ was investigated. It was found that while oscillations redistributed energy amongst the neutrino flavors, the impact on $N_\nu$ was negligible. We have therefore ignored neutrino oscillation effects in our study as we do not have a clear idea why for other  values of $\eta$ and $\theta_W$ the redistribution of neutrino energy would create any larger effect than already determined. Once the relevant neutrino properties are fully understood, the precision of the results we present could possibly be improved by incorporating the effect of neutrino oscillations.

\subsubsection{Scattering Length and Freeze-out Temperature}\label{freeze_out_temp}
The notion of freeze-out temperature is conceptually useful, but within the Boltzmann approach there is no such precise temperature, as the freeze-out process is gradual, with low energy neutrinos freezing-out before high energy ones. Thus  a procedure to determine the freeze-out condition can only be  approximate.  However, the differences that arise through investigating the freeze-out of the three different classes of processes while natural constants are varied help us to understand the results which will be presented below.

To define the freeze-out condition we follow the standard procedure~\cite{kolb}: we compare the distance $L$ traveled by a particle between two scattering processes to the characteristic Universe expansion length $L_H=c/H$. The crossing of the Hubble-length with the neutrino scattering length   produces an estimate of the decoupling temperatures. To obtain the scattering length $L$ we begin with \req{n_div} for the fractional rate of change of comoving particle number
\begin{align}
\frac{\displaystyle\frac{d}{dt}(a^3 n)}{\displaystyle a^3n}=\frac{g_\nu}{2\pi^2n}T^2\int C[f]zdz.
\end{align}
This expression includes both forward and back-reactions, producing  the net change. 

However, we would rather  count the number of interactions.  For that reason, we consider only one direction for the process  and define as the rate of interest
\begin{align}
r\equiv\frac{g_\nu}{2\pi^2n}T^2\int \tilde C[f]zdz
\end{align}
where the forward-reaction operator $\tilde C[f]$ is computed as in \req{coll} except with $F$ replaced by 
\begin{equation}
\tilde F=f_1(p_1)f_2(p_2)f^3(p_3)f^4(p_4).
\end{equation}
If particle type $1$ also participates in the reverse of the reaction $1+2\rightarrow 3+4$ (i.e. it is the same as the final particle $3,4$) then a corresponding term for the reverse reaction must also be added.  The key point is that  we are counting reactions, and not the net particle number change which requires by detailed balance also a negative contribution.

\begin{figure}
\centerline{\includegraphics[width=0.47\columnwidth]{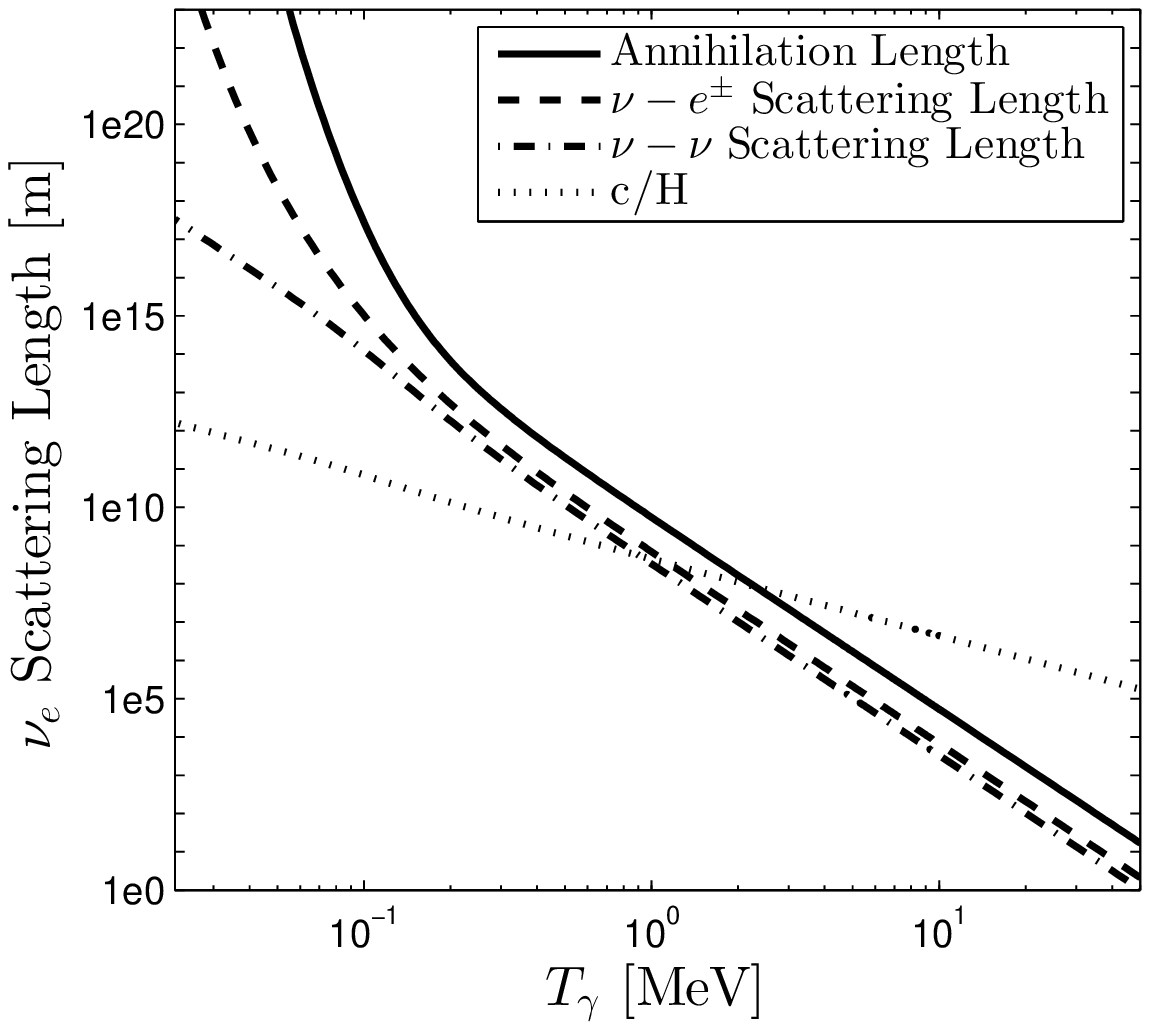}
\includegraphics[width=0.47\columnwidth]{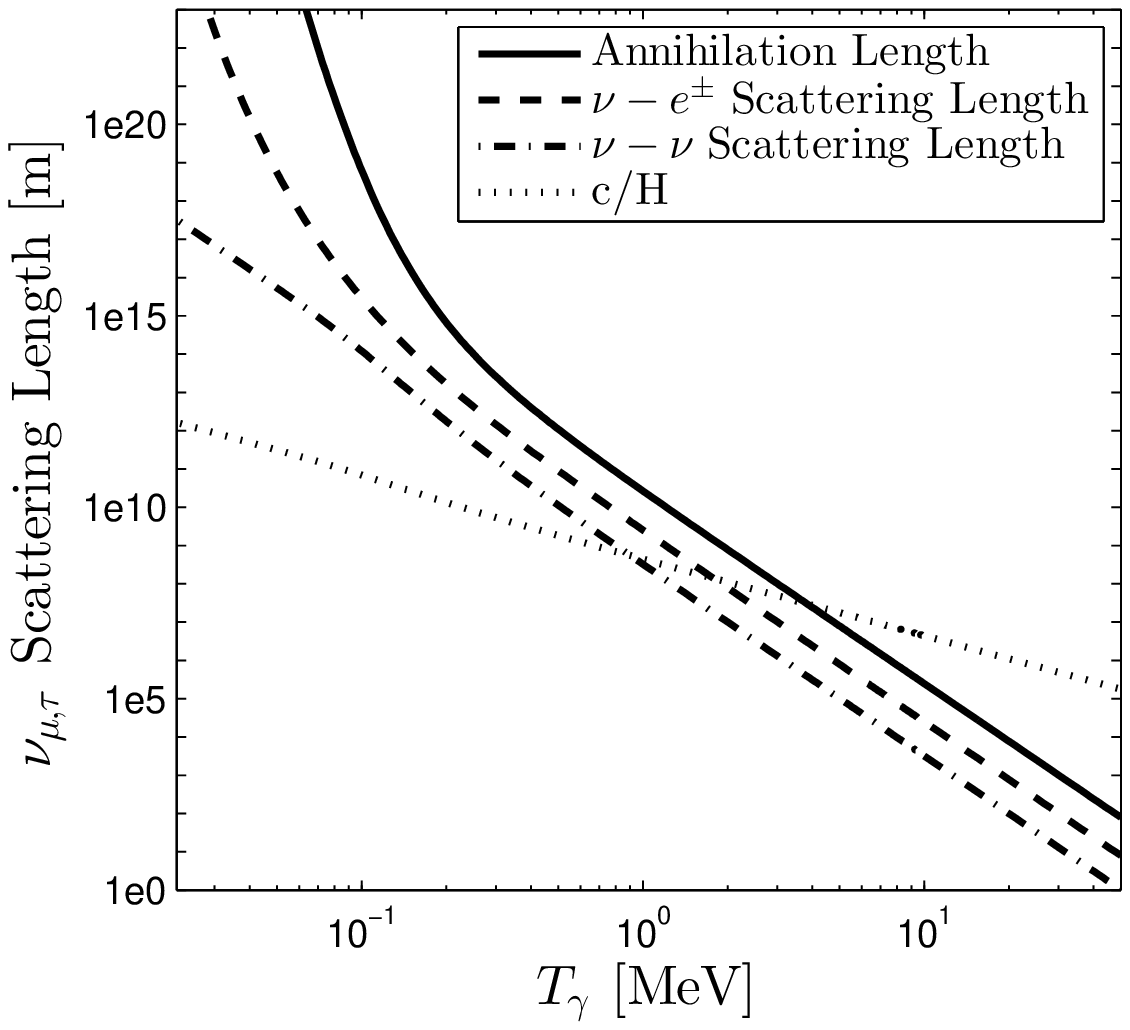}}
\caption{Comparison of Hubble parameter to neutrino scattering length for various types of processes for $\sin^2\theta_W=0.23$ and $\eta=\eta_0$: on left for $\nu_e$ and on right for $\nu_\mu$ and $\nu_\tau$. }\label{fig:scatt_length}
 \end{figure}

Using the average velocity, which for neutrinos is $\bar v=c=1$, we obtain according to this procedure the scattering length
\begin{align}
L&\equiv\frac{\bar v}{r}=\displaystyle{ \frac{\displaystyle\int_0^\infty\frac{z^2dz }{\Upsilon^{-1}e^z+1} }{\displaystyle\sum_i\displaystyle\int_0^\infty \frac{\tilde C_i[f] z^2 dz}{E}}}
\end{align}
where the sum is over the one way scattering operators for the collection of processes of interest. 

Like in  Ref.\cite{kolb}, $L$ can be compared to the Hubble length $L_H=c/H$ and the temperature at which $L=L_H$ we call the freeze-out temperature for that reaction.  Figure \ref{fig:scatt_length} shows $L_H$ together with the scattering length for the three types of neutrino reactions described above, on the left for $\nu_e$ and the right for $\nu_\mu,\nu_\tau$. The flavor dependence  is due to  charge current W-mediated interactions being present only for $\nu_e$. The solid lines in Figure \ref{fig:scatt_length}  corresponds to the chemical freeze-out scattering length, the dashed line corresponds to the kinetic freeze-out scattering length, and  the  dot-dashed line corresponds to re-equilibration  processes within the neutrino fluid. 

Using our Boltzmann equation solver, we have characterized the dependence of neutrino freeze-out temperature on $\sin^2\theta_W$ and $\eta$, shown in figure \ref{fig:freezeoutT}   via the method described above. The left panels show the result for $\nu_e$ and the right for $\nu_\mu,\nu_\tau$.  The SM results corresponding to the crossings in  figure \ref{fig:scatt_length} are read out along the vertical lines in the top two panels. 

\begin{figure}
\centerline{\includegraphics[width=0.47\columnwidth]{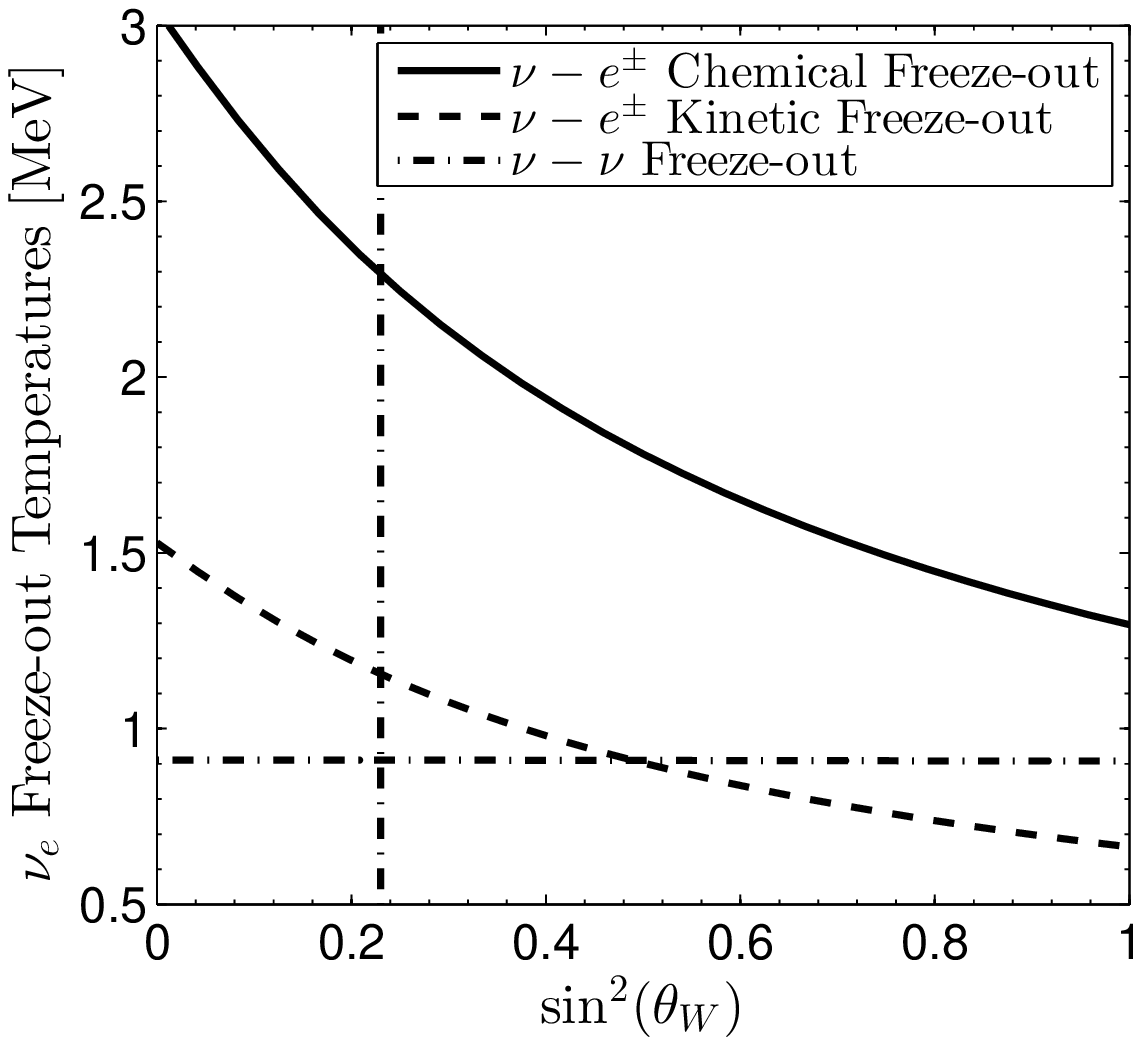}
\hspace{1mm}\includegraphics[width=0.47\columnwidth]{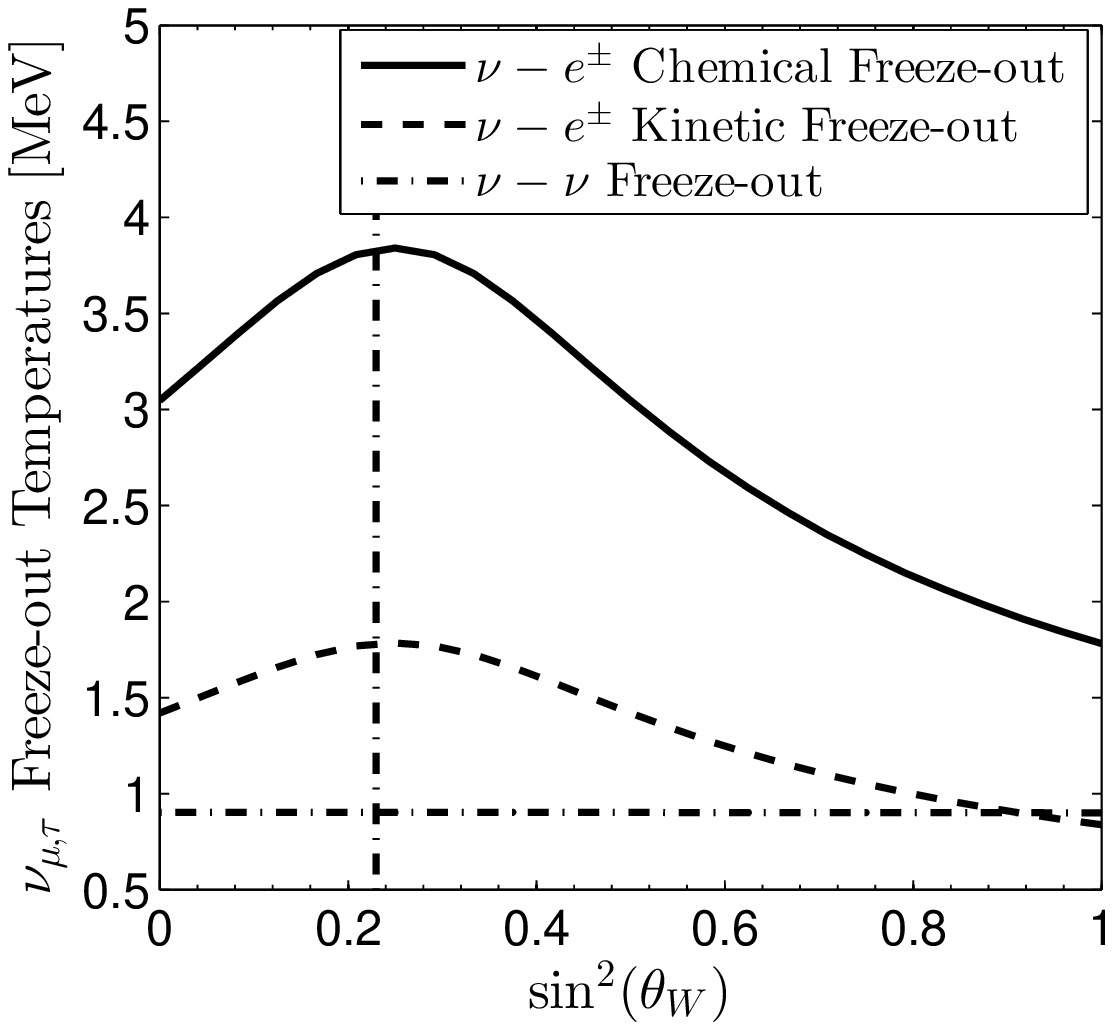}}
\centerline{\includegraphics[width=0.47\columnwidth]{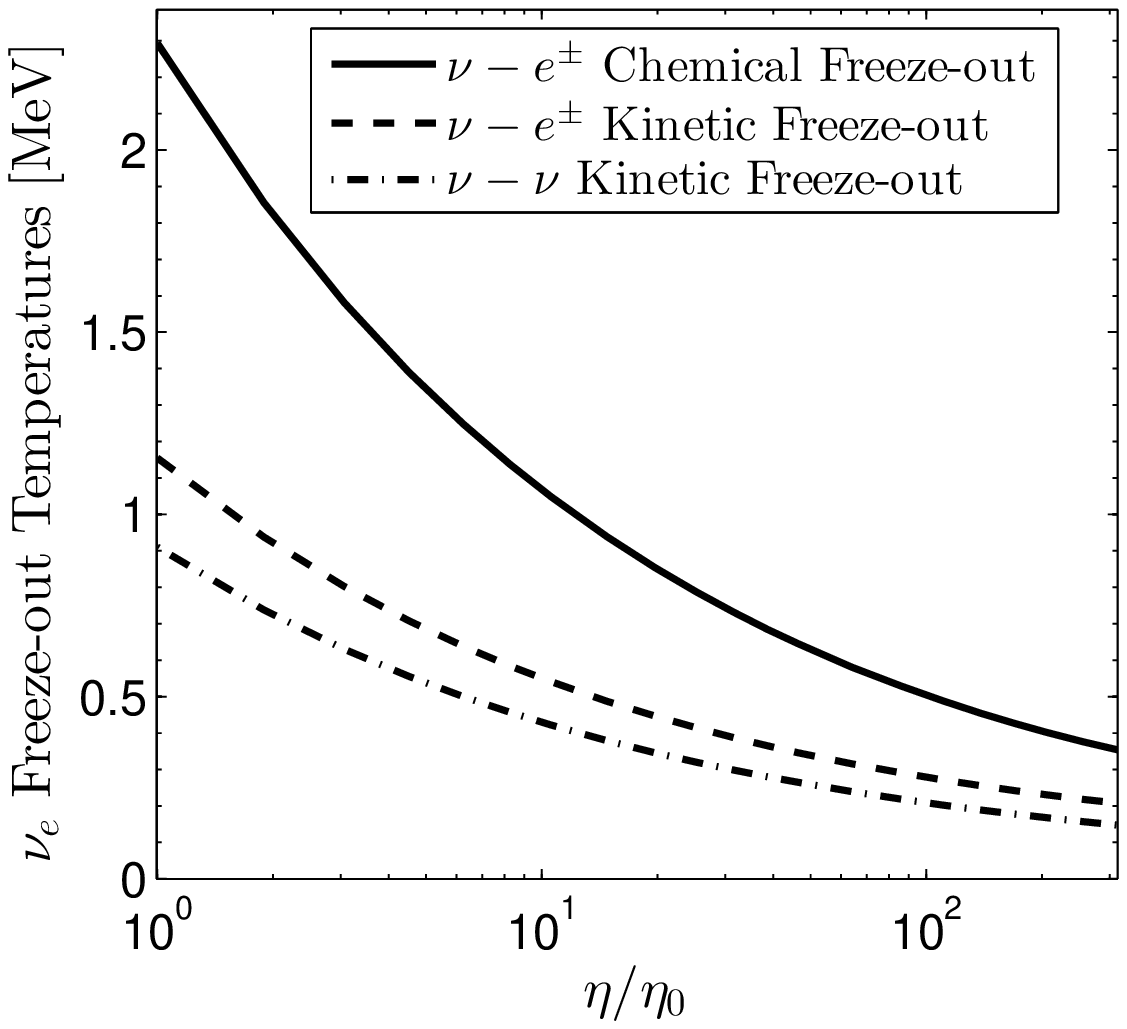}
\hspace{1mm}\includegraphics[width=0.47\columnwidth]{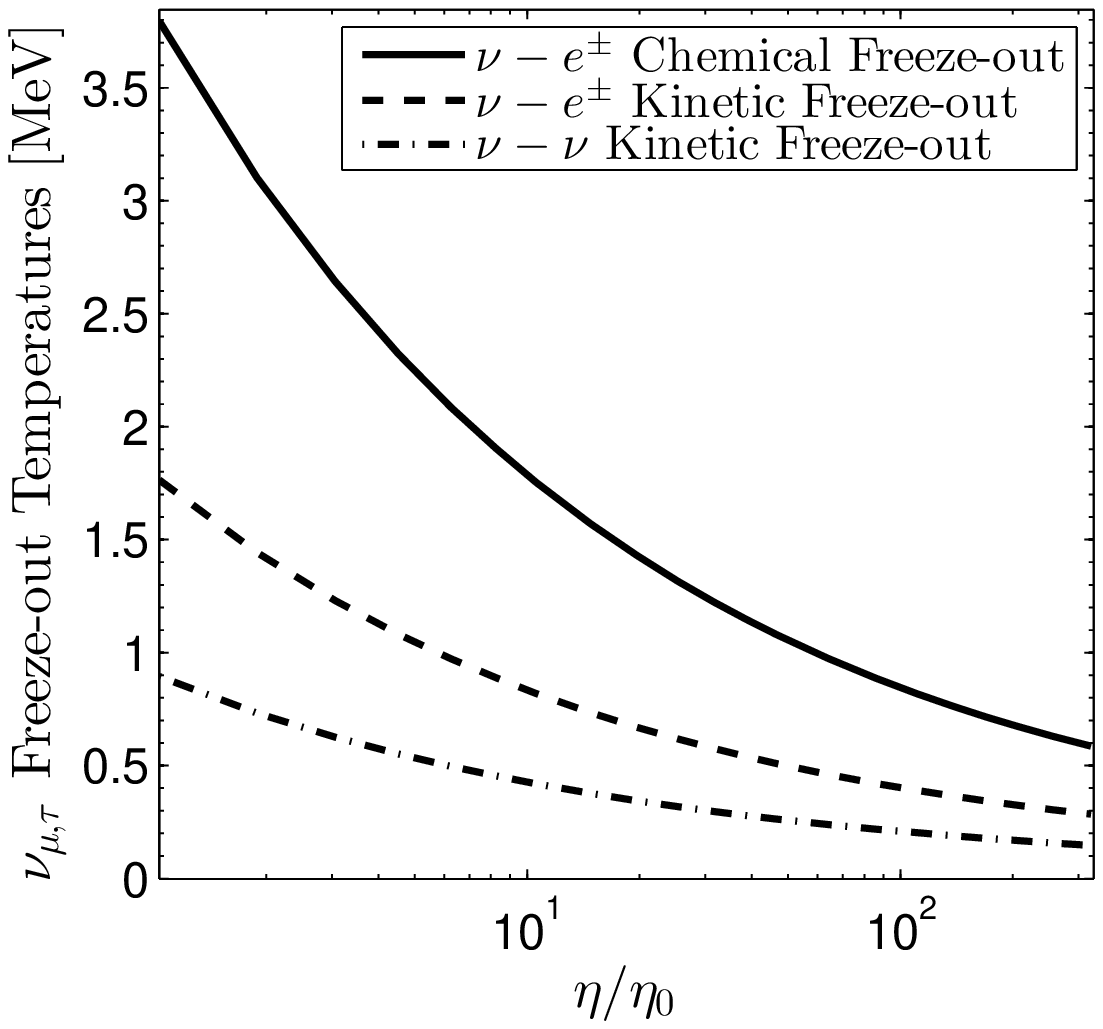}}
\caption{Freeze-out temperatures for electron neutrinos (left) and $\mu$, $\tau$ neutrinos (right) for the three types of processes, see text. Top panels as functions of $\sin^2\theta_W$ for $\eta=\eta_0$, vertical line is $\sin^2\theta_W=0.23$; bottom panels  as a function  of relative change in interaction strength $\eta/\eta_0$ obtained  for $\sin^2\theta_W=0.23$ .}\label{fig:freezeoutT}
 \end{figure}

We see in figure  \ref{fig:freezeoutT} that as a function of $\theta_W$ the behavior of $T_{\nu_e}$ is opposite to that of $T_{\nu_\mu}$ and $T_{\nu_\tau}$ for $\sin^2\theta_W<0.25$, and for $\sin^2\theta_W>0.25$  all neutrino processes tend to decouple at lower temperature with increasing $\sin^2\theta_W$.  Neutrino-neutrino scattering process remains constant, as their matrix elements are independent of Weinberg angle.  An increased coupling strength $\eta$  is equivalent to an increase in $G_F$, resulting in the neutrinos interacting with the $e^\pm$ plasma down to lower temperatures.  Hence the monotonic decreasing behavior of the freeze-out temperature as a function of $\eta/\eta_0$ seen in figure  \ref{fig:freezeoutT} is expected. The SM values are seen at the left margin of the bottom panels. 

As discussed above, neutrino oscillations are not considered in these results.  We expect that incorporating oscillations would lead to a smaller difference between the freeze-out temperatures of the different neutrino flavors, and would likely pull up the drop in $\nu_\mu$, $\nu_e$ freeze-out temperature at small $\sin^2\theta_W$, at least to some degree. We recall that for other observable quantities we discuss in the following, the effect of neutrino oscillations is expected to be negligible~\cite{Mangano2005}.

\subsection{Dependence of $N_\nu$ on Standard Model Parameters}\label{ssec:nnudep}
The main result  of this paper is the  dependence of $N_\nu$ on  the SM parameters   $\sin^2\theta_W$ and $\eta$, \req{eq:etaMp}. These results are shown in the left pane of figure \ref{N_nu_params}, presented as a function of  Weinberg angle $\sin^2 \theta_W $ for $\eta/\eta_0=1,5,10,26$. The effect of an increase in both parameters above the vacuum values superpose  in the parameter range  considered, amplifying the effect and generating a significant increase in  $N_\nu\to 3.5$. 
\begin{figure}
\centerline{\includegraphics[width=0.5\columnwidth]{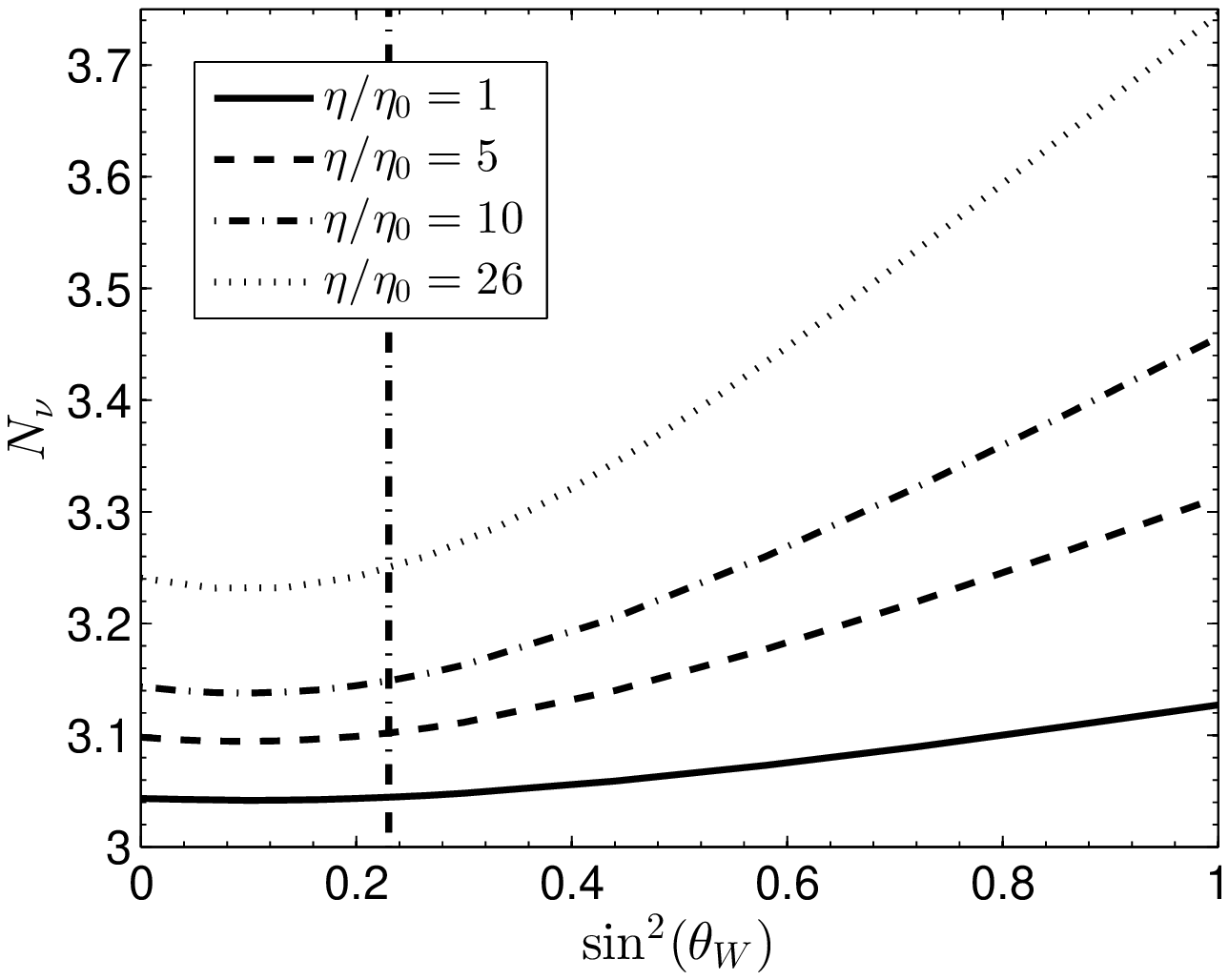}
\hspace{0mm}\includegraphics[width=0.48\columnwidth]{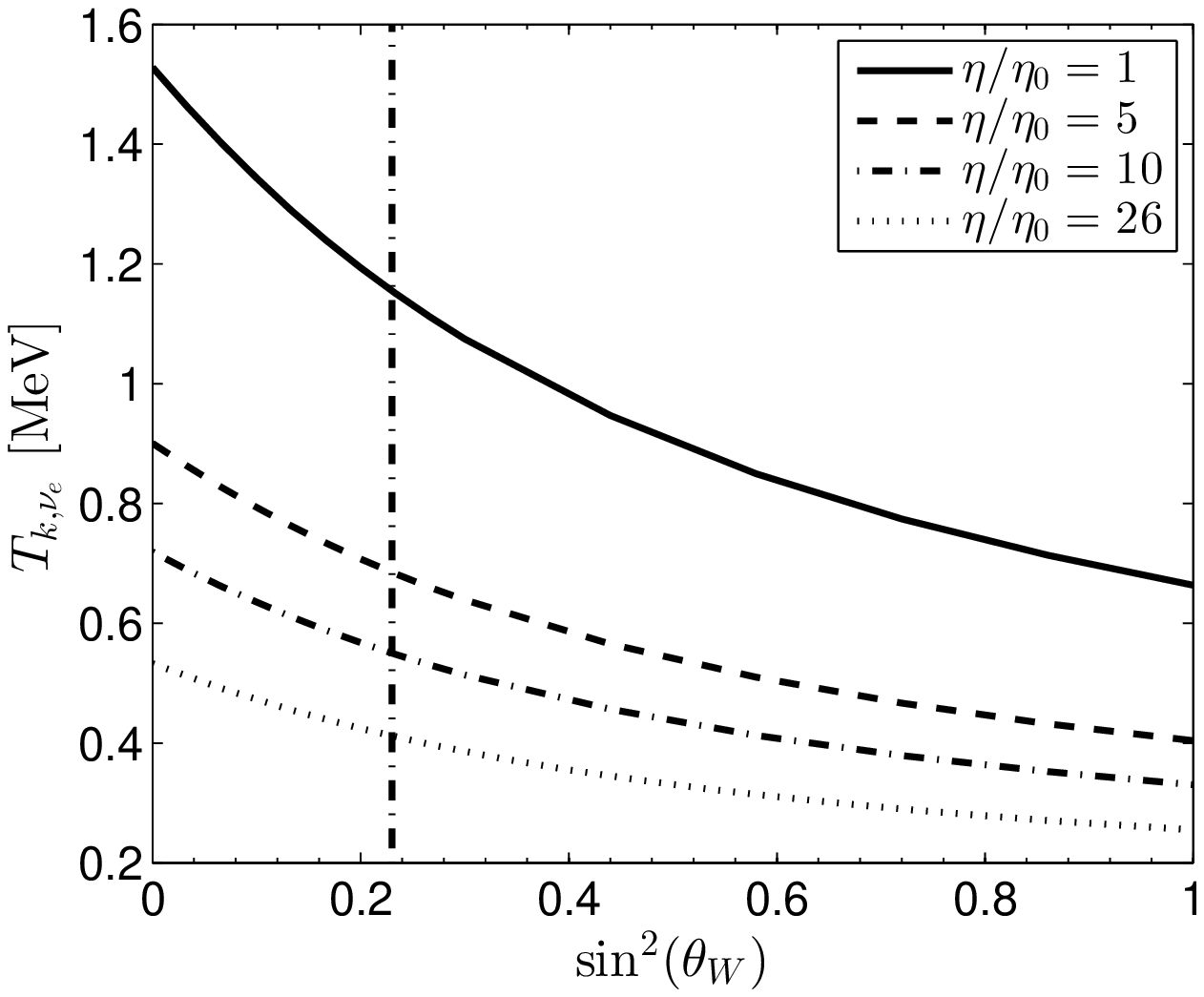}}
\caption{The effective number of neutrinos  $N_\nu$ (left pane), and   the kinetic freeze-out temperature  $T_{k,\nu_e}$ (right pane) as a function of Weinberg angle for  several values of $\eta/\eta_0=1,2,5,10,26$. Vertical line is $\sin^2\theta_W=0.23$.}
\label{N_nu_params}  
 \end{figure}

The last to freeze-out from the kinetic equilibrium is $\nu_e$ and we show the associated value of freeze-out temperature in the right pane of figure \ref{fig:freezeoutT}. Since the freeze-out for present day vacuum value SM parameters occurs well above the electron mass, the reheating effect is normally small.  for  the present day vacuum value of Weinberg angle puts the $\nu_\mu,\nu_\tau$ freeze-out temperature,  $T_{k,\nu_e}=1.2$MeV, seen in the right pane of figure \ref{fig:freezeoutT}. With increasing $\eta$ and $\sin^2\theta_W$ the temperature drops but even for the most extreme cases shown it always remains well above the onset of nucleosynthesis at about $T<0.150$MeV.

We performed a least squares fit of $N_\nu$ over the range $0\leq \sin^2\theta_W\leq 1$, $1\leq \eta/\eta_0\leq 10$ shown in figure \ref{N_nu_params}, obtaining a result with relative error less than $0.2\%$,
\begin{align}
N_\nu=&3.003-0.095\sin^2\theta_W +0.222\sin^4\theta_W  -0.164\sin^6\theta_W \notag\\
+&\sqrt{\frac{\eta}{\eta_0}}\left(0.043+0.011\sin^2\theta_W +0.103\sin^4\theta_W\right).
\end{align}
$N_\nu$ is monotonically increasing in $\eta/\eta_0$ with dominant behavior  scaling as $\sqrt{ \eta/\eta_0}$. Monotonicity is to be expected, as increasing $\eta$ decreases the freeze-out temperature and the longer neutrinos are able to remain coupled to $e^\pm$, the more energy and entropy from annihilation is transferred to neutrinos.

The bounds on $N_\nu$ from the Planck analysis \cite{Planck} can be  used to constrain time or temperature variation of $\sin^2\theta_W$ and $\eta$.  In  Figure \ref{N_nu_domain} the dark (green) color shows the combined range of  variation of natural constants  compatible with CMB+BAO and the light (teal) color shows  the extension in the range of  variation of  natural constants for CMB+$H_0$, both at a $68\%$ confidence level. The dot-dashed line within the dark (green) color  delimits   this latter domain. The dotted line shows the limit of a 5\% change in $N_\nu$.    Any increase in  $\eta/\eta_0$ and/or $\sin^2\theta_W$ moves the value of $N_\nu$ into the domain favored by current experimental results. 

\begin{figure}
\centerline{\includegraphics[width=0.75\columnwidth]{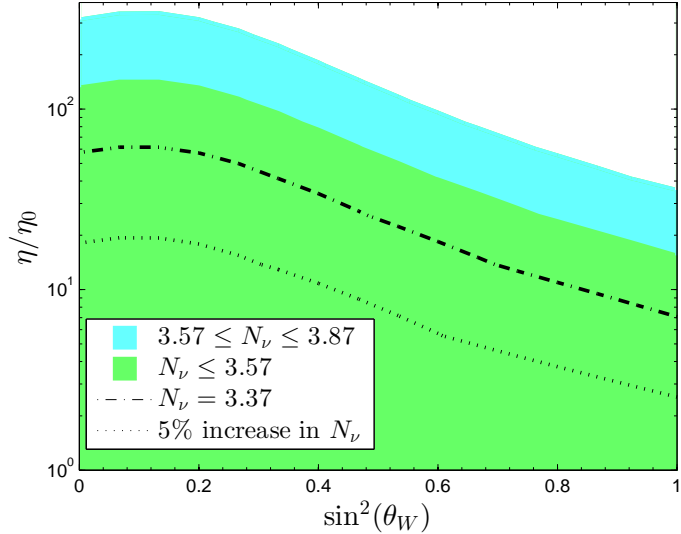}
}
\caption{$N_\nu$ bounds in the $\eta/\eta_0, \sin^2\theta_W$ plane. Dark (green) for $N_\nu\in (3.03,3.57)$ corresponding to Ref.\cite{Planck} CMB+BAO analysis and light(teal) extends the region to $N_\nu<3.87$ i.e. to CMB+$H_0$. Dot-dashed line delimits the 1s.d. lower boundary of the second analysis.}
\label{N_nu_domain}
 \end{figure}

Further parameter study is found in \ref{app:Tups}. In the figures \ref{fig:dist1} and  \ref{fig:dist2} and the data fits Eqs.(\ref{fit1}--\ref{fit4}) we complement the  $N_\nu$ results by showing the variation of the parameters that characterize the neutrino distributions after freeze-out:  the neutrino temperature, shown through the ratio of the reference photon to neutrino temperature $T_\gamma/T_\nu$ separately for $\nu_e$ and $\nu_\mu,\nu_\tau$  and well as the two fugacities $\Upsilon_{\nu_e}$ and   $\Upsilon_{\nu_\mu}=\Upsilon_{\nu_\tau}$.

\section{Connections}\label{sec:disc}
Our study interfaces with two other areas of physics:
\begin{itemize}
\item
Any change of natural constants that would be able to explain a measured variation in $N_\nu$ from SM expectations would need to be made consistent with the ensuing in evolution of the Universe Big Bang nucleosynthesis (BBN). Smoothness of time evolution of the natural constants and the known challenges that beset the BBN results present an interesting avenue of future work which we briefly describe in the following subsection \ref{subsec:BBNconnect}. 
\item
The effective number of neutrinos $N_\nu$ is a characterization of the relativistic energy content in the early Universe, see \req{eq:Nnu}, independent of its source. Thus, even given a conclusive measurement of $N_\nu>3$, there would still remain ambiguity in regard the origin of the effect. Specifically, any light particles that decouple at an earlier epoch  can contribute to the energy content of the invisible Universe. Two potential candidates we describe  below in subsection \ref{subsec:twoNnu}: the sterile neutrino,  and yet to be identified novel nearly `dark' Goldstone Bosons relating to broken symmetries at QGP phase transformation in the early Universe. 
\end{itemize}

\subsection{Connection to Big Bang Nucleosynthesis}\label{subsec:BBNconnect}

Big Bang nucleosynthesis is one of the pillars of modern observational cosmology. It is discussed in comprehensive review articles such as~\cite{boesgaard1985big,tytler2000review,burles2001big,fields2004big,Iocco:2008va} and places strong constraints on the state of the Universe in the temperature range $T_\gamma=100-10$keV. Variation of natural constants that impact the nuclear reaction rates or exansion of the Universe during the BBN era have been investigated, including in particular the time dependence of the neutron to proton mass ratio~\cite{Iocco:2008va,Pospelov:2010hj,Coc:2006sx}, the fine structure constant and deuteron binding energy~\cite{Coc:2006sx}, or Newton's constant (i.e. the Planck mass)~\cite{PhysRevLett.92.171301}. 

Natural constant modifications will not always connect neutrino freeze-out and BBN processes:
\begin{enumerate}
\item The parameter $\sin^2\theta_W$ controlling the relationship between charged and neutral weak currents, and the  interactions of neutrinos within the primordial plasma, does not factor prominantly into BBN, where only combinations that involve the Fermi coupling constant $G_F$ are so far recognized as significant. Therefore, changes in  $\sin^2\theta_W$ which for $\eta>1$ can affect the neutrino freeze-out processes rather strongly, are to best of current knowledge unconstrained by BBN.
\item The neutrino freeze-out remains in a domain of temperature $T>200$keV even for the strongest parameter changes we considered in figure \ref{fig:freezeoutT} right panel. On the other hand,  the  BBN era processes set in for $T<150$keV. Thus if one is willing to accept some fine tuning of the time dependence of natural constants,  there would never be a conflict of neutrino freeze-out modification by natural constants with BBN.
\item In the neutrino freeze-out process all relevant natural constants combine to the one parameter $\eta$, \req{eq:etaMp}.   If both $m_e$ and $G_F$ are varying independently,  effect of their increase can  compensate since it mainly afflicts the neutron abundance. However both effects would be compunded in the neutrino-freeze out process. This can  produce the desired increase in $N_\nu$ without affecting BBN. However, within the SM we expect a strong correlation between $m_e$ and $G_F$. If we assume that minimal SM coupling controls the electron mass, then $\eta\propto 1/v$, since $m_e^3\propto v^3$ while $G_F\propto 1/v^2$, see subsection \ref{ssec:matrix}. The controling scale  $v$ is due to Higgs vacuum structure,  believed to have decoupled from possible modifications near the BBN epoch.
\item Gravity  enters through $\eta\propto 1/\sqrt{G_N}$. As discussed in~\cite{PhysRevLett.92.171301}, the required large decrease in $G_N$  would conflict with BBN unless fine-tuned to phase out before the strong onset of BBN. Thus a combination of neutrino freeze-out process, BBN,  assuming  smoothness of $G_N$ in time and minimal coupling of electrons could set a very strong limit on variation of $G_N$ in the early Universe only a fraction of a second old.
\end{enumerate}
We believe that there is very likely only weak coupling between modifications we consider in the era of neutrino freeze-out and BBN. Thus we could  seek to understand the cosmological value of $N_\nu$ in terms of modifications of natural constants, and only then turn to answer the question how this can be kept consistent with the BBN processes.

In the above discussion we assumed that it is advisable not to perturb BBN. However,  not all is perfectly well with BBN. An outstanding problem is the observed abundance of $^{7}$Li, which is significantly smaller than the prediction of the standard BBN model, see for example Fig.~3 in  Ref.\cite{2011ARNPS..61...47F} and Fig.~5 in   Ref.\cite{cyburt2008update}. The situation with $^{6}$Li also raises concerns but there the looser constraints from reaction processes and after-BBN effects make the larger disagreement less compelling for the much smaller  $^{6}$Li yield. Various approaches to the  $^{7}$Li-problem have been investigated, including non-standard neutron sources~\cite{PhysRevD.90.085018}, nuclear resonances, or dark matter, Ref\cite{2011ARNPS..61...47F} provides comprehensive references on the latter two possibilities. The potential for an explanation of this effect within the context of late neutrino freeze-out modification of natural constants has not been explored, and it is not immediately obvious how this could work. It is possible that delayed decoupling of neutrinos could contributed to some key reaction, but future work is needed before anything definitive can be said.

\subsection{$N_\nu$ from Dark Radiation}\label{subsec:twoNnu}

In this paper, we considered the possibility of modifying $N_\nu$ via neutrinos sharing in  a greater fraction of the entropy of annihilating $e^\pm$, achieved by a change in natural constants.  As $N_\nu$ is only a measure of the relativistic energy density leading up to photon decoupling, a natural alternative mechanism for obtaining $N_\nu>3$ is the introduction of additional, presently not discovered, weakly interacting (effectively) massless particles.  As discussed in Refs.~\cite{Anchordoqui:2011nh,Anchordoqui:2012qu,Blennow:2012de,Steigman:2013yua,Birrell:2014connect}, such particles can contribute fractionally to $N_\nu$ depending on their degeneracy, Bose-Fermi nature, and freeze-out temperature.   For the study of the impact of such dark radiation on BBN see~\cite{Ruchayskiy:2012si,steigman2012neutrinos}.

Of particular relevance could be a so called light sterile neutrino~\cite{Abazajian:2012ys}, possibly the right handed complement to the left handed neutrinos. If such particles exist and freeze-out well before regular neutrinos, their contribution to $N_\nu$ would be subject to dilution by reheating~\cite{Birrell:2014connect} and thus would depend on when precisely they begin free-streaming.

These unknown dark `radiation' particles, as well as neutrinos, could have a mass that is at the scale of the temperature of photon decoupling $T_{\gamma 0}=0.25$ eV, for which an analysis of the Universe density fluctuations akin to Planck~\cite{Planck} would need to be adapted. We have  discussed in Ref.\cite{Birrell:2013_2} a consistent treatment of neutrino mass and $N_\nu$ in the case of a particular type of delayed massive neutrino  freeze-out. This approach is  the same as for dark radiation:  Near to $T_{\gamma 0}=0.25$ eV massive neutrinos are indistinguishable from massive dark radiation, which contributes as an additional particle with reduced contribution to $N_\nu$~\cite{Birrell:2014connect}.

Removing the degeneracy in the interpretation of $N_\nu$ as being due to the decoupling processes of neutrinos, or due to the presence of `dark' particles will naturally depend on other experimental information, such as the contribution to resolving the Li puzzle in BBN or other experimental impacts of dark particles, and of course a contribution from both avenues could be envisioned.

\section{Summary, Discussion and Conclusions}\label{sec:concl}
We have employed a novel spectral method Boltzmann solver and a new procedure for evaluating the Boltzmann scattering integrals in order to characterize the impact of a potential time and/or temperature variation of SM parameters on the effective number of neutrinos. Specifically, we identified a dimensionless combination of $m_e$, $M_p$, and $G_F$, called the interaction strength $\eta$, that, along with the Weinberg angle $\sin^2 \theta_W$, control neutrino freeze-out and the resulting value of the effective number of neutrinos, $N_\nu$.

\subsection{Novel Mathematical Tool}\label{sec:math_tool}
In order to carry this comprehensive study we have developed a novel approach to obtain Boltzmann Equation solutions. Our spectral method, which we call the emergent chemical non-equilibrium method, employs a moving (in Hilbert space) frame, in  which  the orthogonal polynomial basis  dynamically evolves to suit the problem.  Our approach as presented here makes several modifications that both improve its numerical speed and make it better suited to the regime we are investigating, namely the stronger coupling between neutrinos and $e^\pm$ that is obtained when SM parameters are varied, and that lead to an increase in $N_\nu$.  As detailed in the general presentation of the method \cite{Birrell_orthopoly}, the improvements are
\begin{enumerate}
\item  We allow a general time dependence of the effective temperature parameter $T$ i.e. we do not assume redshift temperature scaling $T\propto1/a$ -- this accommodates the effect of reheating.  Without this, the method would be very inefficient in systems with strong reheating,  eventually leading to a failure to converge when the reheating ratio exceeds $2$.
\item We have introduced a chemical non-equilibrium distribution in the weight function i.e. we introduced an evolving, time dependent  $\Upsilon$ which equals $1$ at high temperature, corresponding to chemical equilibrium, and allows for the emergence of chemical non-equilibrium $\Upsilon\neq 1$ during freeze-out.
\item We have introduced an additional factor of $z^2$ to the functional form of the weight as proposed in a different context in Refs.\cite{Wilkening,Wilkening2} which accounts in our approach for the effectively massless neutrino phase space.
\end{enumerate}

The salient feature is that we are letting the fugacity, $\Upsilon$ and temperature $T$  be time dependent and there is no requirement that  $\Upsilon\to 1$.  This should be contrasted with the method used in \cite{Esposito2000,Mangano2002}, which we call the chemical equilibrium method, that studied neutrino freeze-out using a fixed orthogonal polynomial basis generated by the chemical equilibrium weight and without the $z^2$ factor.  The chemical equilibrium method also assumes a particular temperature scaling  $Ta(t)=$Const.  In other words, the neutrino momenta are scaled by $1/a(t)$ instead of a dynamical effective temperature $T(t)$ as in our method. Such a method is effective for the weak reheating found for SM vacuum parameters, but it becomes less inefficient and eventually fails to converge as the reheating ratio increases.

Due to the inclusion of the neutrino phase space $z^2$ factor in the weight, and facilitated by the near thermal shape of the distribution, only two  modes corresponding to  $T$ and $\Upsilon$ are required to capture the energy density and number density of the neutrino distribution.  In comparison, the chemical equilibrium method, because it lacks the $z^2$ factor, requires a minimum of four modes.  We discussed how further important savings in computation time are arrived at by making the integrands of the collision integrals smooth functions. Overall, the speed up of solutions is at level 20 times or more.

\subsection{Primordial Variation of Natural Constants}
The question which we answer in this paper is: What neutrino decoupling in the early Universe can tell us about the values of natural constants when the Universe was about 1 second old and at an ambient temperature near to 1.2 MeV (14 billion degrees K). Our results were presented assuming that the Universe contains no other effectively massless particles but the three left handed neutrinos and corresponding, three right handed anti-neutrinos. 

We found that near to the physical value of the Weinberg angle  $\sin^2 \theta_W\simeq 0.23$ the effect of changing $\sin^2\theta_W$ on the decoupling of neutrinos is small. Thus as seen in Figure \ref{N_nu_params}  the dominant variance is due to the change  in the coupling strength $\eta/\eta_0$, \req{eq:etaMp}  and \req{eta_def}. The dotted line in  Figure \ref{N_nu_domain} shows that in order to achieve a change in $N_\nu$ at the level of up to 5\%, that is  $N_\nu\lesssim 3.2$,  both $\sin^2 \theta_W$ and $\eta/\eta_0$ must change significantly, with e.g. $\eta$ increasing by an order of magnitude.

Let us review what an increase in the strength parameter $\eta$ by factor 10 means, looking case by case on all the natural constant contributions as if each were responsible for the entire change:
\begin{itemize}
\item 
Considering that  $\eta\propto M_p\propto G_N^{-1/2}$ this translates into a decrease  in the strength of Gravity at neutrino freeze-out by a factor 100.  This effect would need to become much smaller by the time the age of the Universe is 1000 times longer (1s compared to 10 min) for Big Bang nucleosynthesis to be unaffected. This presumably means that, conversely, as we go further back in time we would need gravity to continue to rapidly become very much weaker yet. In models of emergent gravity we can  imagine a  `melting' of gravity in the hot primordial Universe. Whether such a model can be realized will be a topic for future consideration. The attractive aspect of Gravity weakening rapidly with increasing temperature is that for  exponentially disappearing $G_N\to 0$ as $t\to 0$ and/or $T\to \infty$ the dynamics can be arranged to be similar to an inflationary  Universe.
\item 
Since $\eta\propto m_e^3$, electron mass would need to go up `only' by factor 2.15 . Compared to all other particles the electron mass has an anomalously  low value. Appearance of a mechanism just when $T\simeq m_e$ that `restores' the electron mass to where intuition would like it to be, a few MeV, arising from  the systematics of other Yukawa Higgs coupling $g_{Ye}$ compared to the Yukawa coupling of other charged light particles, where $m_e= g_{Ye} v $ seems to us also  a possible scenario. Interestingly,   laboratory limits for these conditions could be attainable in the foreseeable future.
\item
Since $\eta\propto G_F^2\propto 1/v^4$  we would need to find a mechanism that would decrease the vacuum value $v_0\simeq 246$ GeV by factor 1.8 already at temperature $T\simeq m_e$.  Allowing three powers of $v$ to cancel by using the Higgs minimal coupling formula for electron mass  we need to change $v$ by an order of magnitude near to $T\simeq m_e$. This appears impossible.
\end{itemize}
While ideas justifying strong variation of $\eta$ can be developed as two of the above three cases argue, a model for temperature or time dependence of  $\sin^2 \theta_W$ seems at this time without a theoretical anchor point, mainly so since we do not have a valid grand unified theoretical framework in which the electro-weak mixing or equivalently the masses $M_W, M_Z$ would be anchored.

{\bf  To conclude:} The  explanation of $N_\nu>3.05$ in terms of variation of natural constants that we have presented comprises  speculative and beyond the standard model ideas akin, in this aspect, to the alternative explanation based on new dark `radiation' particles. In order to achieve an increase in $N_\nu$ the change in natural constants must cause, through a  delay in neutrino freeze-out, a greater participation of neutrinos in reheating during $e^\pm$ annihilation.  We believe that the study here presented shows a viable  mechanism  capable of influencing $N_\nu$, and thus merits further investigation. In particular, reconciliation with the following BBN epoch  will help to estimate limits on  variation  in the early Universe of the two fundamental parameters controlling  $N_\nu$: $\eta$ (see \req{eq:etaMp}) and $\sin^2\theta_W$, the latter parameter in principle remaining unconstrained by BBN and could freely evolve as long as it reaches the present day measured value.

\subsection*{Acknowledgments}
This work has been supported by a grant from the U.S. Department of Energy, DE-FG02-04ER41318
and was conducted with Government support under and awarded by DoD, Air Force Office of Scientific Research, National Defense Science and Engineering Graduate (NDSEG) Fellowship, 32 CFR 168a.

\appendix
\section{Inducing Volume Forms on Submanifolds}\label{app:delta}
Given a Riemannian manifold $(M,g)$ with volume form $dV_g$ and a  hypersurface $S$, the standard Riemannian hypersurface area form, $dA_g$ is defined on $S$ as the volume form of the pullback metric tensor on $S$.  Equivalently, it can be computed as
\begin{equation}
dA_g=i_v dV_g
\end{equation}
where $v$ is a unit normal vector to $S$ and $i_v$ denotes interior product (i.e. contraction) of the antisymmetric tensor $dV_g$ with the vector $v$.  

We take a moment to describe the properties of the interior product that are relevant for our purposes. The interior product, $i_X\omega$, is linear in both the vector $X$ and the form $\omega$ and for a one form (i.e. dual vector) $\tau$, $i_X\tau$ is just the usual contraction of a vector and dual vector.  On higher degree forms the interior product is characterized by the relation
\begin{equation}
i_X(\omega\wedge\tau)=(i_X\omega)\wedge\tau+(-1)^k\omega\wedge(i_X\tau)
\end{equation}
where $\omega$ is a $k$-form and $\tau$ is an $l$-form.  In particular in a coordinate system $x^i$, contracting the coordinate volume element with a coordinate vector $\partial_{x^i}$ is straightforward
\begin{equation}
i_{\partial_{x^i}}dx^1...dx^n=(-1)^{i-1}dx^1...dx^{i-1}dx^{i+1}...dx^n
\end{equation}
where we omit the wedge product signs. In the following we will only be concerned with the results up to sign (i.e. with the density defined by a volume form).  Iterated contractions with the vectors $X_j$ will be denoted by $i_{(X_1,...,X_m)}$.  As we are only concerned with the result up to sign, the order in which we contract is irrelevant.

The Riemannian method of inducing volume measures extends to submanifolds of codimension greater than one as well as to semi-Riemannian manifolds, as long as the metric restricted to the submanifold is non-degenerate, by contracting with an orthonormal basis for the normal vectors. However, there are many situations where one would like to define a natural volume form on a submanifold that is induced by a volume form in the ambient space, but where the above method is inapplicable, such as defining a natural volume form on the light cone or other more complicated degenerate submanifolds in relativity. In this appendix, we will describe a method for inducing volume forms on regular level sets of a function that is applicable in cases where there is no metric structure and show its relation to more widely used semi-Riemannian case. 

Let $M$, $N$ be smooth manifolds, $c$ be a regular value of a smooth function $F:M\rightarrow N$, and $\Omega^M$ and $\Omega^N$ be volume forms on $M$ and $N$ respectively.  Using this data, we will be able to induce a natural volume form on the level set $F^{-1}(c)$.  The absence of a metric on $M$ is made up for by the additional information that the function $F$ and volume form $\Omega^N$ on $N$ provide. The following theorem makes our definition precises and proves the existence and uniqueness of the induced volume form.

\begin{theorem}\label{induced_vol_form}
Let $M$, $N$ be $m$ (resp. $n$)-dimensional smooth manifolds with volume forms $\Omega^M$ (resp. $\Omega^N$). Let $F:M\rightarrow N$ be smooth and $c$ be a regular value.  Then there is a unique volume form $\omega$  (also denoted $\omega^M$) on $F^{-1}(c)$ such that $\omega_x=i_{(v_1,...,v_n)}\Omega^M_x$ whenever $v_i\in T_xM$ are such that 
\begin{equation}\label{unit_volume}
\Omega^N(F_*v_1,...,F_* v_n)=1.
\end{equation}
We call $\omega$ the {\bf volume form induced by $F:(M,\Omega^M)\rightarrow (N,\Omega^N)$}.
\end{theorem}
\begin{proof}
$F_*$ is onto $T_{F(x)}N$ for any $x\in F^{-1}(c)$.  Hence there exists $\{v_i\}_1^n\subset T_xM$ such that 
\begin{equation}
\Omega^N(F_*v_1,...,F_* v_n)=1.
\end{equation}
  In particular, $F_* v_i$ is a basis for $T_{F(x)} N$.  Define $\omega_x=i_{(v_1,...,v_n)}\Omega_x$. This is obviously a nonzero $m-n$ form on $T_xF^{-1}(c)$ for each $x\in F^{-1}(c)$.  We must show that this definition is independent of the choice of $v_i$ and the result is smooth.\\

 Suppose $F_*v_i$ and $F_*w_i$ both satisfy \req{unit_volume}.  Then $F_*v_i=A_i^jF_*w_j$ for $A\in SL(n)$. Therefore $v_i-A_i^jw_j\in \ker F_{*x}$.  This implies
\begin{equation}
i_{(v_1,...,v_n)}\Omega^M_x=\Omega^M_x(A_1^{j_1}w_{j_1},...,A_n^{j_n}w_{j_n},\cdot)
\end{equation}
since the terms involving $\ker F_*$ will vanish on $T_x F^{-1}(c)=\ker F_{*x}$.  Therefore
\begin{align}\label{ind_of_v_proof}
i_{(v_1,...,v_n)}\Omega^M_x&=A_1^{j_1}...A_n^{j_n}\Omega^M_x(w_{j_1},...,w_{j_n},\cdot)\\
&=\sum_{\sigma\in S_m} \pi(\sigma)A_1^{\sigma(1)}...A_n^{\sigma(n)}\Omega^M_x(w_1,...,w_n,\cdot)\\
&=\det(A)i_{(w_1,...,w_n)}\Omega^M_x\\
&=i_{(w_1,...,w_n)}\Omega^M_x.
\end{align}
This proves that $\omega$ is independent of the choice of $v_i$.  If we can show $\omega$ is smooth then we are done.  We will do better than this by proving that for any  $v_i\in T_xM$ the following holds
\begin{equation}
i_{(v_1,...,v_n)}\Omega^M_x=\Omega^N(F_*v_1,...,F_*v_n)\omega_x.
\end{equation}
To see this, take $w_i$ satisfying \req{unit_volume}.  Then $F_*v_i=A_i^j F_*w_j$. This determinant can be computed from
\begin{align}
\Omega^N(F_*v_1,...,F_*v_n)=\det(A)\Omega^N(F_*w_1,...,F_*W_n)=\det(A).
\end{align}
 Therefore, the same computation as \req{ind_of_v_proof} gives
\begin{align}
i_{(v_1,...,v_n)}\Omega^M_x=\det(A)\omega_x=\Omega^N(F_*v_1,...,F_*v_n)\omega_x
\end{align}
as desired.  To prove that $\omega$ is smooth, take a smooth basis of vector fields $\{V_i\}_1^m$ in a neighborhood of $x$.  After relabeling, we can assume $\{F_*V_i\}_1^n$ are linearly independent at $F(x)$ and hence, by continuity, they are linearly independent at $F(y)$ for all $y$ in some neighborhood of $x$.  In that neighborhood, $\Omega^N(F_*V_1,...,F_*V_n)$ is non-vanishing and therefore
\begin{equation}
\omega=(\Omega^N(F_*V_1,...,F_*V_n))^{-1}i_{(V_1,...,V_n)}\Omega
\end{equation} 
which is smooth.
\end{proof}

\begin{corollary}\label{induced_vol_eq}
For any  $v_i\in T_xM$ the following holds
\begin{equation}
i_{(v_1,...,v_n)}\Omega^M_x=\Omega^N(F_*v_1,...,F_*v_n)\omega_x.
\end{equation}
\end{corollary}

\begin{corollary}
If $\phi:M\rightarrow\mathbb{R}$ is smooth and $c$ is a regular value then by equipping $\mathbb{R}$ with its canonical volume form we have 
\begin{equation}
\omega_x=i_v\Omega^M_x
\end{equation}
where $v\in T_xM$ is any vector satisfying $d\phi(v)=1$.
\end{corollary}

A coarea formula can be proved for the induced volume forms.
\begin{theorem}[Coarea formula]\label{vol_form_coarea}
Let $M$ be a smooth manifold with volume form $\Omega^M$, $N$ a smooth manifold with volume form $\Omega^N$ and $F:M\rightarrow N$ be a smooth map.  If $F_*$ is surjective at a.e. $x\in M$ then for $f\in L^1(\Omega^M)\bigcup L^+(M)$
\begin{equation}\label{coarea_formula}
\int_Mf(x) \Omega^M(dx)=\int_{N}\int_{F^{-1}(z)} f(y)\omega^M_z(dy) \Omega^N(dz)
\end{equation}
where $\omega^M_z$ is the volume form induced on $F^{-1}(z)$ as in theorem \ref{induced_vol_form}.
\end{theorem}

 The induced measure defined above allows for a coordinate independent definition of a delta function supported on a regular level set.  Such an object is of great use in performing calculations in relativistic phase space in a coordinate independent manner. 
\begin{definition}
Motivated by the coarea formula, we define the composition of the {\bf Dirac delta function} supported on $c\in N$ with a smooth map $F:M\rightarrow N$ such that $c$ is a regular value of $F$ by
\begin{equation}\label{delta_def}
 \delta_c(F(x))\Omega^M \equiv \omega^M
\end{equation}
on $F^{-1}(c)$.    For $f\in L^1(\omega^M)$ we will write 
\begin{equation}
\int_M f(x)\delta_c(F(x))\Omega^M(dx)
\end{equation} 
in place of 
\begin{equation}
\int_{F^{-1}(c)} f(x) \omega^M(dx).
\end{equation}
\end{definition}

It is useful to translate the induced volume element into a form that is more readily applicable to computations in coordinates.  Choose arbitrary coordinates $y^i$ on $N$ and write $\Omega^N=h^N(y) dy^n$. Choose coordinates $x^i$ on $M$ such that $F^{-1}(c)$ is the coordinate slice
\begin{equation}
F^{-1}(c)=\{x:x^1=...=x^n=0\}
\end{equation}
and write $\Omega^M=h^M(x)dx^m$. The coordinate vector fields $\partial_{x^i}$ are transverse to $F^{-1}(c)$ and so
\begin{equation}
\Omega^N(F_*\partial_{x^1},...,F_*\partial_{x^n})=h^N(F(x))\det \left(\frac{\partial F^i}{\partial x^j}\right)_{i,j=1..n}
\end{equation}
and
\begin{equation}
i_{(\partial_{x^1},...,\partial_{x^n})}\Omega^M=h^M(x) dx^{n+1}...dx^m.
\end{equation}
Therefore we obtain
\begin{equation}\label{vol_form_coords}
\omega_x=\frac{h^M(x)}{h^N(F(x))}\det \left(\frac{\partial F^i}{\partial x^j}\right)^{-1}_{i,j=1..n}dx^{n+1}...dx^m.
\end{equation}

Using \req{vol_form_coords}, along with the coordinates described there, we can (at least locally) write the integral with respect to the delta function in the more readily usable form
\begin{equation}\label{delta_integral_coords}
\int_M f(x)\delta_c(F(x))\Omega^M=\int_{F^{-1}(c)} f(x)\frac{h^M(x)}{h^N(F(x))}\bigg|\det \left(\frac{\partial F^i}{\partial x^j}\right)^{-1}\bigg|dx^{n+1}...dx^m.
\end{equation}
The absolute value comes from the fact that we use $\delta_c(F(x))\Omega^M$ to define the orientation on $F^{-1}(c)$.

\section{Electron and Neutrino Collision Integrals}\label{app:nu_matrix_elements}
\subsection{$\nu\nu\rightarrow\nu\nu$ }
Using \req{Mandelstam}, the matrix elements for neutrino neutrino scattering can be simplified to
\begin{align}
\label{TA002}
S|\mathcal{M}|^2=C(p_1\cdot p_2)(p_3\cdot p_4)=C\frac{s^2}{4}
\end{align}
 where the coefficient $C$ is given in table \ref{table:nu_nu_coeff}.

\begin{table}[h]
\centering 
\begin{tabular}{|c|c|}
\hline
Process &$C$ \\
\hline
$\nu_i+\nu_i\rightarrow\nu_i+\nu_i,\hspace{2mm} i\in\{e,\mu,\tau\}$& $64 G_F^2$\\
\hline
$\nu_i+\nu_j\rightarrow\nu_i+\nu_j,\hspace{2mm} i\neq j, \hspace{1mm} i,j\in\{e,\mu,\tau\}$& $32 G_F^2$\\
\hline
\end{tabular}
\caption{Matrix element coefficients for neutrino neutrino scattering processes.}
\label{table:nu_nu_coeff}
\end{table}
From here we obtain
{\small
\begin{align}
M_{\nu\nu\rightarrow\nu\nu}=&\frac{C}{256(2\pi)^5 }\int_{s_0}^\infty\!\!\!\! s^2\!\!\int_0^\infty \!\!\!\!\int_{-1}^1\! G_{12}(p^0,-pz)dz \int_{-1}^1\!G_{34}\left(p^0,-py\right) dy\frac{ p^2}{p^0}dpds.
\end{align}
}
Therefore, as we claimed above, $M_{\nu\nu\rightarrow\nu\nu}$ can be written in a form that requires the numerical evaluation of only three iterated integrals, but not quite as a three dimensional integral. If we want to emphasize the role of $C$ then we write $M_{\nu\nu\rightarrow\nu\nu}(C)$.  Note that if one scales $p$ and $s$ by the appropriate powers of $T$ in order to convert to dimensionless variables, one obtains a prefactor of $T^8$.

\subsection{$\nu\bar\nu\rightarrow \nu\bar\nu$ }
Using \req{Mandelstam}, the matrix elements for neutrino anti-neutrino scattering can be simplified to
\begin{align}
S|\mathcal{M}|^2=C\left(\frac{s+t}{2}\right)^2
\end{align}
where the coefficient $C$ is given in table \ref{table:nu_nubar_coeff}.

\begin{table}[h]
\centering 
\begin{tabular}{|c|c|}
\hline
Process &$C$  \\
\hline
$\nu_i+\bar\nu_i\rightarrow\nu_i+\bar\nu_i,\hspace{2mm} i\in\{e,\mu,\tau\}$& $128 G_F^2$\\
\hline
$\nu_i+\bar\nu_i\rightarrow\nu_j+\bar\nu_j,\hspace{2mm} i\neq j, \hspace{1mm} i,j\in\{e,\mu,\tau\}$& $32 G_F^2$\\
\hline
$\nu_i+\bar\nu_j\rightarrow\nu_i+\bar\nu_j,\hspace{2mm} i\neq j, \hspace{1mm} i,j\in\{e,\mu,\tau\}$& $32 G_F^2$\\
\hline
\end{tabular}
\caption{Matrix element coefficients for neutrino neutrino scattering processes.}
\label{table:nu_nubar_coeff}
\end{table}
Using this we find
 \begin{align}
\int_0^{2\pi} S |\mathcal{M}|^2 (s,t(\cos(\psi)\sqrt{1-y^2}\sqrt{1-z^2}+yz))d\psi=&\frac{\pi C}{16} s^2(3+4 yz-y^2-z^2+3y^2z^2)\notag\\
\equiv&\frac{\pi C}{16} s^2q(y,z),
\end{align}

\begin{align}
M_{\nu\bar\nu\rightarrow\nu\bar\nu}=&\frac{C}{2048(2\pi)^5 }T^8\!\!\!\int_0^\infty\!\!\!\int_0^\infty\!\!\!\! {s}^2\left[\int_{-1}^1\int_{-1}^1q(y,z){G}_{34}(p^0,-{p} y) {G}_{12}(p^0,-{p} z)dydz\right]\!\frac{{p}^2}{{p}^0}d{p}d{s}.
\end{align}
 Again, by converting to dimensionless variables we see that this scales with $T^8$. If we want to emphasize the role of $C$ then we write $M_{\nu\bar\nu\rightarrow\nu\bar\nu}(C)$. Note that due to the polynomial form of the matrix element integral, the double integral in brackets breaks into a linear combination of products of one dimensional integrals, meaning that the nesting of integrals is only three deep.

\subsection{$\nu\bar{\nu}\rightarrow e^+e^-$}\label{nu_nubar_int}
Using \req{Mandelstam}, the matrix elements for neutrino anti-neutrino annihilation into $e^\pm$ can be simplified to
\begin{align}
S|\mathcal{M}|^2=A\left(\frac{s+t-m_e^2}{2}\right)^2+B\left(\frac{m_e^2-t}{2}\right)^2+Cm_e^2\frac{s}{2}
\end{align}
where the coefficients $A,B,C$ are given in table \ref{table:nu_nubar_ee_coeff}.

\begin{table}[h]
\centering 
\begin{tabular}{|c|c|c|c|}
\hline
Process &$A$&$B$&$C$  \\
\hline
$\nu_e+\bar\nu_e\rightarrow e^++e^-$&$128G_F^2g_L^2$&$128G_F^2g_R^2$&$128G_F^2g_Lg_R$\\
\hline
$\nu_i+\bar\nu_i\rightarrow e^++e^-,\hspace{2mm} i\in\{\mu,\tau\}$&$128G_F^2\tilde g_L^2$&$128G_F^2g_R^2$&$128G_F^2\tilde g_Lg_R$\\
\hline
\end{tabular}
\caption{Matrix element coefficients for neutrino neutrino annihilation into $e^\pm$.}
\label{table:nu_nubar_ee_coeff}
\end{table}

  The integral of each of these terms is
 \begin{align}
&\int_0^{2\pi}\frac{(s+t(\psi)-m_e^2)^2}{4}d\psi=\frac{\pi}{16}s(3s-4m_e^2)+\frac{\pi}{4}s^{3/2}\sqrt{s-4m_e^2}yz\\
&-\frac{\pi}{16}s(s-4m_e^2)(y^2+z^2)+\frac{3\pi}{16}s(s-4m_e^2)y^2z^2,\notag\\
&\int_0^{2\pi} \frac{(m_e^2-t(\psi))^2}{4}d\psi=\frac{\pi}{16}s(3s-4m_e^2)-\frac{\pi}{4}s^{3/2}\sqrt{s-4m_e^2}yz\\
&-\frac{\pi}{16}s(s-4m_e^2)(y^2+z^2)+\frac{3\pi}{16}s(s-4m_e^2)y^2z^2,\\
&\int_0^{2\pi} m_e^2\frac{s}{2} d\psi=\pi m_e^2s.
\end{align}
Therefore 
\small
\begin{align}
\int_0^{2\pi} S |\mathcal{M}|^2 (s,t(\psi))d\psi=&\frac{\pi}{16}s\left[3s(A+B)+4m_e^2(4C-A-B)\right]+\frac{\pi}{4}s^{3/2}\sqrt{s-4m_e^2}(A-B)yz\notag\\
&-\frac{\pi}{16}s(s-4m_e^2)(A+B)(y^2+z^2)+\frac{3\pi}{16}s(s-4m_e^2)(A+B)y^2z^2\notag\\
\equiv& \pi q(m_e,s,y,z).
\end{align}
\begin{align}
M_{\nu\bar\nu\rightarrow e^+e^-}=&\frac{1}{128(2\pi)^5 }\int_{4m_e^2}^\infty\int_0^\infty\!\!\!\sqrt{1-4m_e^2/s}\left[\int_{-1}^1\int_{-1}^1q(m_e,s,y,z)G_{34}(p^0,-(\sqrt{1-4m_e^2/s})p y)\right.\notag\\
&\hspace{68mm}\times G_{12}(p^0,-p z)dydz\bigg]\frac{ p^2}{p^0}dpds,
\end{align}
\normalsize
By scaling $s$, $p$, and $m_e$ by the appropriate powers of $T$ we again obtain a prefactor of $T^8$.  If we want to emphasize the role of $A,B,C$ then we write $M_{\nu\bar\nu\rightarrow e^+e^-}(A,B,C)$.  Note that this expression is linear in $(A,B,C)\in\mathbb{R}^3$. Also note that, under the assumptions that the distributions of $e^+$ and $e^-$ are the same (i.e. ignoring the small matter anti-matter asymmetry), the $G_{ij}$ terms that contain the product of $e^\pm$ distributions are even functions. Hence the term involving the integral of $yz$ vanishes by antisymmetry.

\subsection{ $\nu e^\pm\rightarrow \nu e^\pm$}
 Using \req{Mandelstam}, the matrix elements for neutrino $e^\pm$ scattering can be simplified to
\begin{align}
\label{TA002_1}
S|\mathcal{M}|^2=A\left(\frac{s-m_e^2}{2}\right)^2+B\left(\frac{s+t-m_e^2}{2}\right)^2+Cm_e^2\frac{t}{2}
\end{align}
 where the coefficients $A,B,C$ are given in table \ref{table:nu_e_coeff}.

\begin{table}[h]
\centering 
\begin{tabular}{|c|c|c|c|}
\hline
Process &$A$&$B$&$C$  \\
\hline
$\nu_e+e^-\rightarrow \nu_e+e^-$&$128G_F^2g_L^2$&$128G_F^2g_R^2$&$128G_F^2g_Lg_R$\\
\hline
$\nu_i+e^-\rightarrow \nu_i+e^-,\hspace{2mm} i\in\{\mu,\tau\}$&$128G_F^2\tilde g_L^2$&$128G_F^2g_R^2$&$128G_F^2\tilde g_Lg_R$\\
\hline
$\nu_e+e^+\rightarrow \nu_e+e^+$&$128G_F^2g_R^2$&$128G_F^2g_L^2$&$128G_F^2g_Lg_R$\\
\hline
$\nu_i+e^+\rightarrow \nu_i+e^+,\hspace{2mm} i\in\{\mu,\tau\}$&$128G_F^2 g_R^2$&$128G_F^2\tilde g_L^2$&$128G_F^2\tilde g_Lg_R$\\
\hline
\end{tabular}
\caption{Matrix element coefficients for neutrino $e^\pm$ scattering.}
\label{table:nu_e_coeff}
\end{table}

  The integral of each of these terms is
 \begin{align}
&\int_0^{2\pi}\frac{(s-m_e^2)^2}{4} d\psi=\pi\frac{(s-m_e^2)^2}{2},\\
&\int_0^{2\pi}\frac{(s+t(\psi)-m_e^2)^2}{4}d\psi=\frac{\pi}{16s^2}(s-m_e^2)^2(3m_e^4+2m_e^2s+3s^2)+\frac{\pi}{4s^2}(s-m_e^2)^3(s+m_e^2)yz,\notag\\
&-\frac{\pi}{16s^2}(s-m_e^2)^4(y^2+z^2)+\frac{3\pi}{16s^2}(s-m_e^2)^4y^2z^2,\\
&\int_0^{2\pi} m_e^2\frac{t(\psi)}{2}d\psi=-\frac{\pi}{2s}m_e^2(s-m_e^2)^2(1-yz).
\end{align}
Therefore we have
\begin{align}
\int_0^{2\pi} S |\mathcal{M}|^2 (s,t(\psi))d\psi=&\pi\left[\frac{A}{2}+\frac{B}{16s^2}(3m_e^4+2m_e^2s+3s^2)-\frac{C}{2s}m_e^2\right](s-m_e^2)^2\notag\\
&+\pi\left[\frac{B}{4s^2}(s-m_e^2)(s+m_e^2)+\frac{C}{2s}m_e^2\right](s-m_e^2)^2yz\notag\\
&-B\frac{\pi}{16s^2}(s-m_e^2)^4(y^2+z^2)+B\frac{3\pi}{16s^2}(s-m_e^2)^4y^2z^2\notag\\
\equiv& \pi q(m_e,s,y,z)
\end{align}
and
\begin{align}\label{matrix_elem_int}
r=\rp =&\frac{s-m_e^2}{\sqrt{s}},\hspace{2mm} q^0=(\qp )^0=-\frac{m_e^2}{\sqrt{s}},\hspace{2mm} \delta=\frac{p}{\sqrt{s}},\hspace{2mm}  \alpha=\frac{p^0}{\sqrt{s}}.
\end{align}

\begin{align}
M_{\nu e\rightarrow\nu e}=&\frac{1}{128(2\pi)^5 }\int_{m_e^2}^\infty\!\int_0^\infty (1-m_e^2/s)^2\left(\int_{-1}^1 \int_{-1}^1 q(m_e,s,y,z) G_{34}\left(p^0,(\qp )^0\alpha-\rp \delta y\right)\right.\notag\\
&\hspace{60mm}\times  G_{12}(p^0,q^0\alpha-r\delta z)dydz\bigg)\frac{ p^2}{p^0}dpds.
\end{align}
As above, after scaling $s$, $p$, and $m_e$ by the appropriate powers of $T$  we obtain a prefactor of $T^8$.  If we want to emphasize the role of $A,B,C$ then we write $M_{\nu e\rightarrow\nu e}(A,B,C)$.  Note that this expression is also linear in $(A,B,C)\in\mathbb{R}^3$.
\subsection{Total Collision Integral}
We now give the total collision integrals for neutrinos.    In the following, we indicate which distributions are used in each of the four types of scattering integrals discussed above by using the appropriate subscripts. For example, to compute $M_{\nu_e\bar\nu_\mu\rightarrow\nu_e\bar\nu_\mu}$  we set $G_{1,2}=\hat\psi_jf^1f^2$, $G_{3,4}=f_3f_4$, $f_1=\hat\psi_j f_{\nu_e}$, $f_3=f_{\nu_e}$, and $f_2=f_4=f_{\bar\nu_\mu}$ in the expression for $M_{\nu\bar\nu\rightarrow\nu\bar\nu}$ from section \ref{nu_nubar_int} and then, to include the reverse direction of the process, we must {\emph subtract}  the analogous expression whose only difference is $G_{1,2}=\hat\psi_j f_1f_2$, $G_{3,4}=f^3f^4$.
With this notation the collision integral for $\nu_e$ is
\begin{align}\label{M_tot}
M_{\nu_e}=&[M_{\nu_e\nu_e\rightarrow\nu_e\nu_e}+M_{\nu_e\nu_\mu\rightarrow\nu_e\nu_\mu}+M_{\nu_e\nu_\tau\rightarrow\nu_e\nu_\tau}]\\
&+[M_{\nu_e\bar\nu_e\rightarrow\nu_e\bar\nu_e}+M_{\nu_e\bar\nu_e\rightarrow\nu_\mu\bar\nu_\mu}+M_{\nu_e\bar\nu_e\rightarrow\nu_\tau\bar\nu_\tau}+M_{\nu_e\bar\nu_\mu\rightarrow\nu_e\bar\nu_\mu}+M_{\nu_e\bar\nu_\tau\rightarrow\nu_e\bar\nu_\tau}]\notag\\
&+M_{\nu_e\bar\nu_e\rightarrow e^+e^-}+[M_{\nu_e e^-\rightarrow\nu_e e^-}+M_{\nu_e e^+\rightarrow\nu_e e^+}]\notag.
\end{align}

Symmetry among the interactions implies that the distributions of $\nu_\mu$ and $\nu_\tau$ are equal.  We also neglect the extremely small matter anti-matter asymmetry and so we take the distribution of each particle to be equal to that of the corresponding antiparticle.  Therefore there are only three independent distributions, $f_{\nu_e}$, $f_{\nu_\mu}$, and $f_e$ and so we can combine some of the terms in \req{M_tot} to obtain
\begin{align}
M_{\nu_e}=&M_{\nu_e\nu_e\rightarrow\nu_e\nu_e}(64G_F^2)+M_{\nu_e\nu_\mu\rightarrow\nu_e\nu_\mu}(2\times 32 G_F^2)+M_{\nu_e\bar\nu_e\rightarrow\nu_e\bar\nu_e}(128G_F^2)  \\
&+M_{\nu_e\bar\nu_e\rightarrow\nu_\mu\bar\nu_\mu}(2\times 32G_F^2)+M_{\nu_e\bar\nu_\mu\rightarrow\nu_e\bar\nu_\mu}(2\times 32 G_F^2)\notag\\
&+M_{\nu_e\bar\nu_e\rightarrow e^+e^-}(128G_F^2g_L^2,128G_F^2g_R^2,128G_F^2g_Lg_R)\notag\\
&+M_{\nu_e e\rightarrow\nu_e e}(128 G_F^2( g_L^2+g_R^2),128 G_F^2 (g_L^2+ g_R^2),256G_F^2g_Lg_R)\notag.
\end{align}
Introducing one more piece of notation, we use a subscript $k$ to denote the orthogonal polynomial basis element that multiplies $f_1$ or $f^1$ in the inner product.  The inner product of the $k$th basis element with the total scattering operator for electron neutrinos is therefore 
\begin{align}
R_k=&2\pi^2T^{-3} M_{k,\nu_e}.
\end{align}
Under these same assumptions and conventions, the total collision integral for the combined $\nu_\mu$, $\nu_\tau$ distribution (which we label $\nu_\mu$) is
\begin{align}
M_{\nu_\mu}=&M_{\nu_\mu\nu_\mu\rightarrow\nu_\mu\nu_\mu}(64G_F^2+32G_F^2)+M_{\nu_\mu\nu_e\rightarrow\nu_\mu\nu_e}(32 G_F^2)\notag\\
&+M_{\nu_\mu\bar\nu_\mu\rightarrow\nu_\mu\bar\nu_\mu}(128G_F^2+32G_F^2+32G_F^2) \notag \\
&+M_{\nu_\mu\bar\nu_\mu\rightarrow\nu_e\bar\nu_e}(32G_F^2)+M_{\nu_\mu\bar\nu_e\rightarrow\nu_\mu\bar\nu_e}( 32 G_F^2)\notag\\
&+M_{\nu_\mu\bar\nu_\mu\rightarrow e^+e^-}(128G_F^2\tilde g_L^2,128G_F^2g_R^2,128G_F^2\tilde g_Lg_R)\notag\\
&+M_{\nu_\mu e\rightarrow\nu_\mu e}(128 G_F^2(\tilde  g_L^2+g_R^2),128 G_F^2 (\tilde g_L^2+ g_R^2),256G_F^2\tilde g_Lg_R),\\
R_k=&2\pi^2T^{-3} M_{k,\nu_\mu}.
\end{align}


\subsection{Conservation Laws and Scattering Integrals}
For some processes, some of the $R_k$'s vanish exactly.  As we now show, this is an expression of various conservation laws. First consider processes in which $f_1=f_3$ and $f_2=f_4$, such as $e^\pm \nu\rightarrow e^\pm\nu$. Since $m_1=m_3$ and $m_2=m_4$ we have $r=\rp $, $q^0=(\qp )^0$.  The scattering terms are all two dimensional integrals of some function of $s$ and $p$ multiplied by 
\small
\begin{align}
I_k\equiv&\int_{-1}^1 \left[\int_{-1}^1\left(\int_0^{2\pi}S|\mathcal{M}|^2 (s,t(\cos(\psi)\sqrt{1-y^2}\sqrt{1-z^2}+yz))d\psi\right) f_1(h_1(y))f_2(h_2(y)) dy\right] \notag\\
&\hspace{26mm}\times f_k^1(h_1(z))f^2(h_2(z))dz\\
&-\int_{-1}^1 \left[\int_{-1}^1\left(\int_0^{2\pi}S|\mathcal{M}|^2 (s,t(\cos(\psi)\sqrt{1-y^2}\sqrt{1-z^2}+yz))d\psi\right) f^1(h_1(y))f^2(h_2(y)) dy\right] \notag\\
&\hspace{26mm}\times f_{1,k}(h_1(z))f_2(h_2(z))dz\\
h_1(y)=&(p^0+(\qp )^0\alpha-\rp \delta y)/2,\hspace{2mm} h_2(y)=(p^0-q^0\alpha+r\delta y)/2, \hspace{2mm} f_{1,k}=\hat\psi_k f_1,\hspace{2mm} f^1_k=\hat\psi_k f^1.
\end{align}
\normalsize
For $k=0$, $\hat\psi_0$ is constant.  After factoring it out of $I_k$, the result is obviously zero and so $R_0=0$.  

We further specialize to a distribution scattering from itself i.e. $f_1=f_2=f_3=f_4$.  Since $m_1=m_2$ and $m_3=m_4$ we have $q^0=(\qp )^0=0$ and
\begin{equation}
h_1(y)=(p^0-\rp \delta y)/2,\hspace{2mm} h_2(y)=(p^0+r\delta y)/2.
\end{equation}
 By the above, we know that $R_0=0$.  $\hat\psi_1$ appears in $I_1$ in the form $\hat\psi_1(h_1(z))$, a degree one polynomial in $z$.  Therefore $R_1$ is a sum of two terms, one which comes from the degree zero part and one from the degree one part.  The former is zero, again by the above reasoning.  Therefore, to show that $R_1=0$ we need only show $I_1=0$, except with $\hat\psi_1(h_1(z))$ replaced by  $z$.  Since $h_1(-y)=h_2(y)$, changing variables  $y\rightarrow -y$ and $z\rightarrow -z$ in the following shows that this term is equal to its own negative, and hence is zero
\begin{align}
&\int_{-1}^1 \left[\int_{-1}^1\left(\int_0^{2\pi}S|\mathcal{M}|^2 (s,t(\cos(\psi)\sqrt{1-y^2}\sqrt{1-z^2}+yz))d\psi\right) f_1(h_1(y))f_1(h_2(y)) dy\right] \notag\\
&\hspace{26mm}\times  zf^1(h_1(z))f^1(h_2(z))dz\\
&-\int_{-1}^1 \left[\int_{-1}^1\left(\int_0^{2\pi}S|\mathcal{M}|^2 (s,t(\cos(\psi)\sqrt{1-y^2}\sqrt{1-z^2}+yz))d\psi\right) f^1(h_1(y))f^1(h_2(y)) dy\right] \notag\\
&\hspace{26mm}\times zf_{1}(h_1(z))f_1(h_2(z))dz.
\end{align}
We note that the corresponding scattering integrals do not vanish for the chemical equilibrium spectral method employed in \cite{Esposito2000,Mangano2002}.  This is another advantage of the method outlined in section \ref{sec:boltz_solve}.  Further differences are discussed  in section \ref{sec:math_tool}.

Finally, we point out how the vanishing of these inner products is a reflection of certain conservation laws. From \req{n_div}, \req{collision_integrals}, and the fact that $\hat\psi_0,\hat\psi_1$ span the space of polynomials of degree $\leq 1$, we have the following expressions for the change in number density and energy density of a massless particle
\begin{align}
\frac{1}{a^3} \frac{d}{dt}(a^3n)=&\frac{g_p}{2\pi^2}\int \frac{1}{E}C[f]p^2dp=c_0 R_0,\\
\frac{1}{a^4}\frac{d}{dt}(a^4\rho)=&\frac{g_p}{2\pi^2}\int C[f] p^2dp=d_0R_0+d_1R_1
\end{align}
for some $c_0,d_0,d_1$. Therefore, the vanishing of $R_0$ is equivalent to conservation of comoving particle number.  The vanishing of $R_0$ and $R_1$ implies $\rho\propto 1/a^4$ i.e. that the reduction in energy density is due entirely to redshift; energy is not lost from the distribution due to scattering.  These findings match the situations above where we found one or both of $R_0=0$, $R_1=0$.  $R_0$ vanishes for scattering processes that exchange momentum but don't change particle number.  Both $R_0$ and $R_1$ vanished for a distribution scattering from itself and in such a process one expects that no energy is lost from the distribution by scattering, it is only redistributed among the particles corresponding to that distribution.

\section{Temperature ratios and fugacities}\label{app:Tups}
We complement the results presented in section \ref{ssec:nnudep} with  photon to neutrino temperature ratios $ T_\gamma / T_{\nu_e}, T_\gamma / T_{\nu_\mu}= T_\gamma / T_{\nu_\tau} $, and the neutrino fugacities, $\Upsilon_{\nu_e}, \Upsilon_{\nu_\mu}=\Upsilon_{\nu_\tau}$, both results are shown in figures \ref{fig:dist1} and \ref{fig:dist2}, varying only one of the two  parameters. 

\begin{figure}
\centerline{\includegraphics[width=0.50\columnwidth]{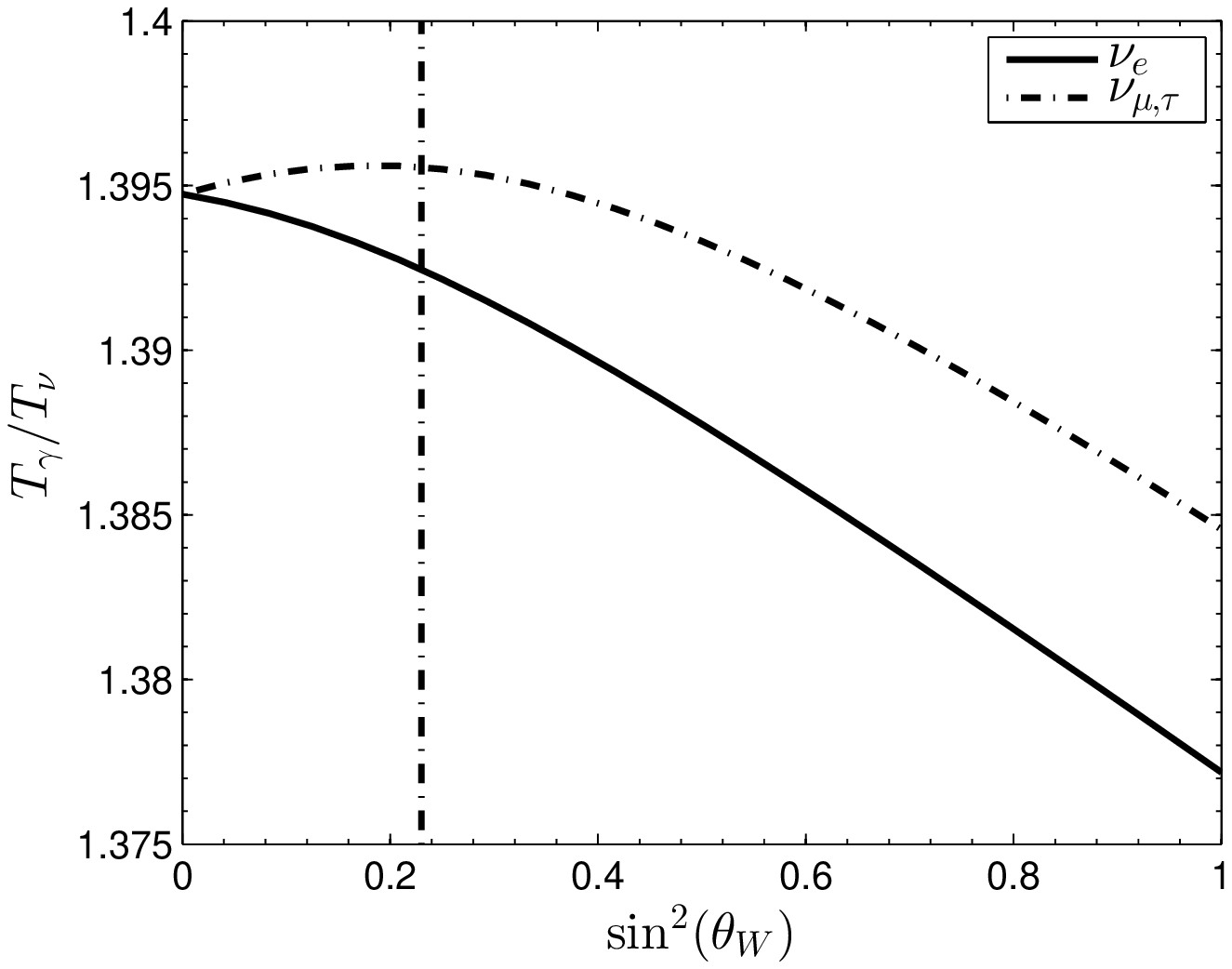}
\hspace{-5mm}\includegraphics[width=0.50\columnwidth]{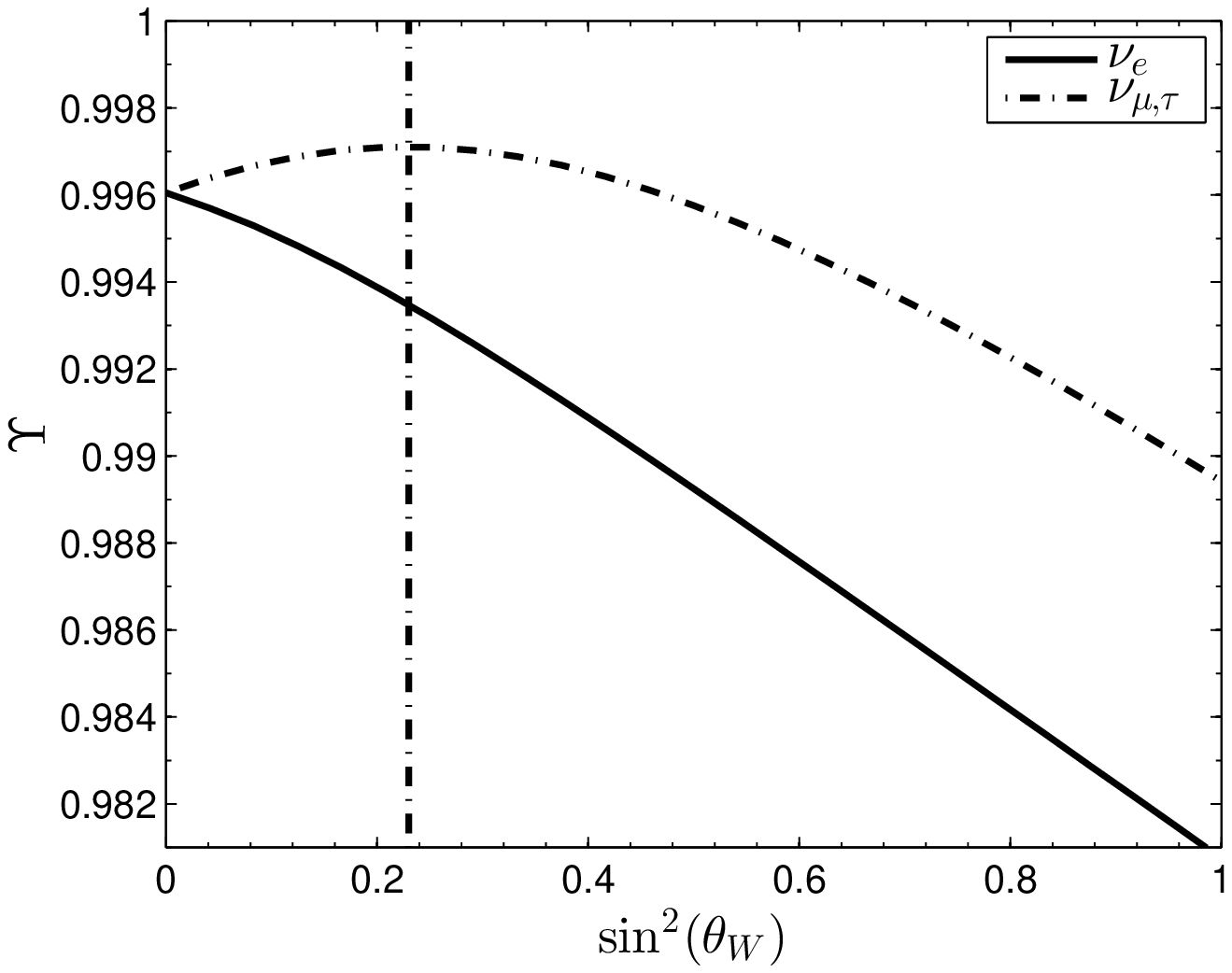}}
\caption{ Photon-neutrino  $\nu_\mu,\nu_\tau$ temperature ratios (left) and neutrino  fugacities (right), as functions of Weinberg angle for $\eta=\eta_0$. }\label{fig:dist1}
\end{figure}
\begin{figure}
\centerline{\includegraphics[width=0.50\columnwidth]{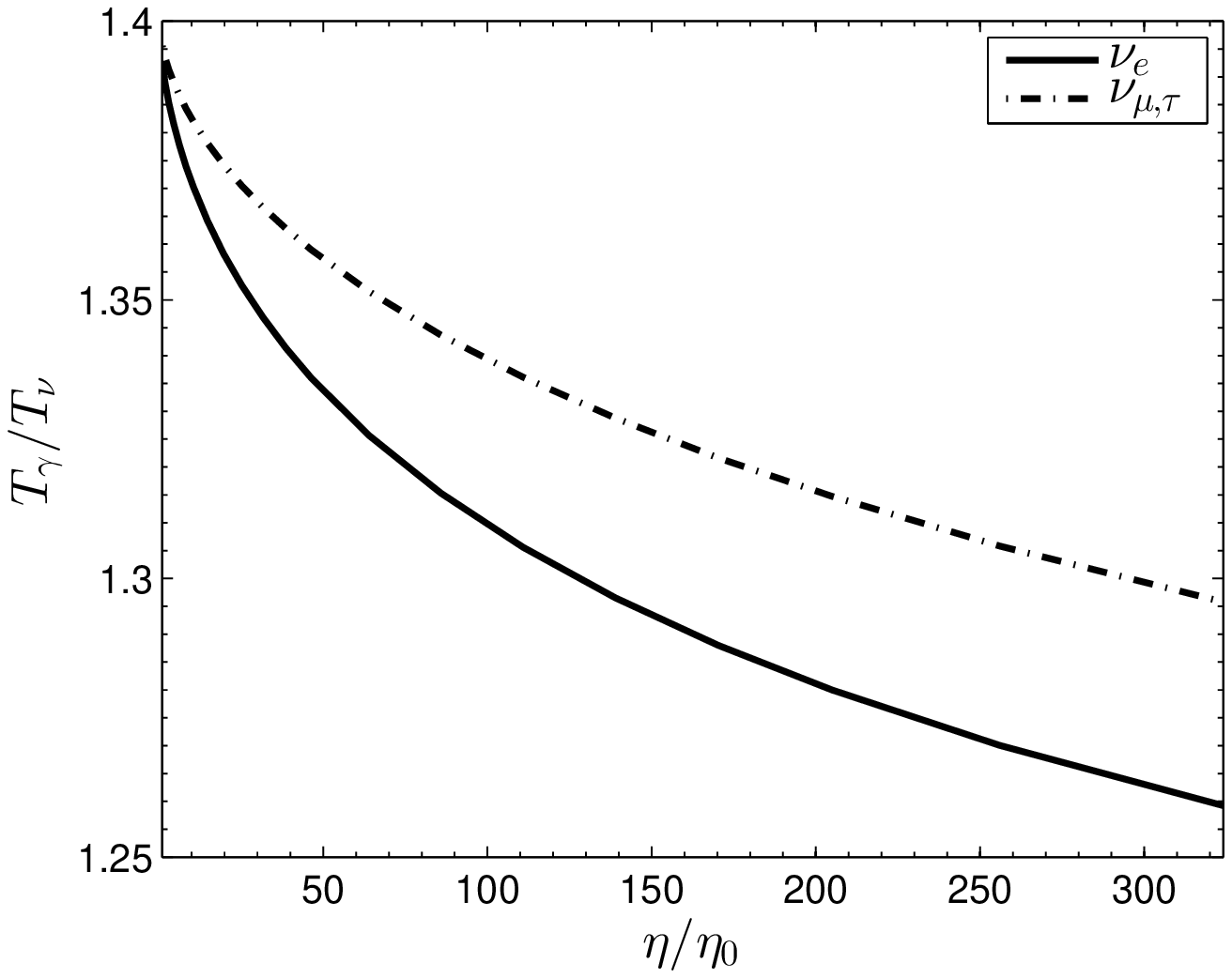}
\hspace{-5mm}\includegraphics[width=0.50\columnwidth]{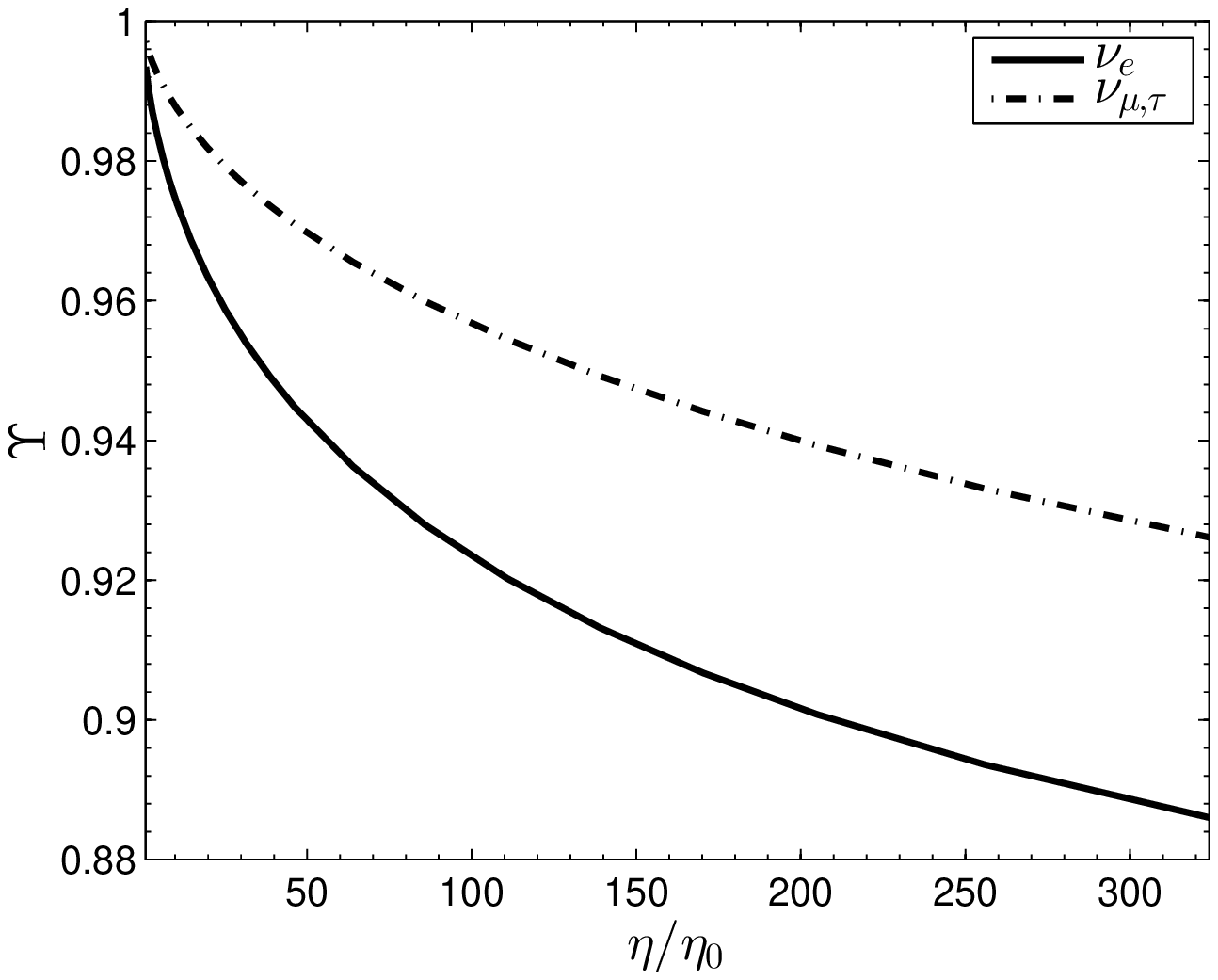}}
\caption{ Photon-neutrino  $\nu_\mu,\nu_\tau$ temperature ratios (left) and neutrino  fugacities (right), as functions of relative interaction strength $\eta/\eta_0$ for $\sin^2\theta_W=0.23$. Vertical line is $\sin^2\theta_W=0.23$.}
\label{fig:dist2}
 \end{figure}

We further show least squares fits to all these quantities for the range $0\leq \sin^2\theta_W\leq 1$, $1\leq \eta/\eta_0\leq 10$ with relative error less than $0.2\%$  
\begin{align}
\frac{T_\gamma}{T_{\nu_\mu}}=&1.401+0.015x-0.040x^2+0.029x^3-0.0065y+0.0040xy-0.017x^2y, \label{fit1}\\
\Upsilon_{\nu_e}=&1.001+0.011x-0.024x^2+0.013x^3-0.005y-0.016xy+0.0006x^2y,\label{fit2}\\ 
\frac{T_\gamma}{T_{\nu_e}}=&1.401+0.015x-0.034x^2+0.021x^3-0.0066y-0.015xy-0.0045x^2y,\label{fit3}\\
\Upsilon_{\nu_\mu}=&1.001+0.011x-0.032x^2+0.023x^3-0.0052y+0.0057xy-0.014x^2y.\label{fit4}
\end{align}
where
\begin{equation}
x\equiv \sin^2 \theta_W ,\qquad
y\equiv  \sqrt{\frac{\eta}{\eta_0}}.
\end{equation}
As mentioned in section \ref{sec:process_types}, neutrino oscillations are neglected in these results.  Presumably, incorporating this effect would lead to a closer match between the fugacities and temperature ratios of the different neutrino flavors.


\bibliography{refs}
\bibliographystyle{utphys_jr}

\end{document}